\documentclass{article}

\providecommand{\neuripsoption}{preprint}
\usepackage[\neuripsoption]{neurips_2026}

\usepackage[utf8]{inputenc}
\usepackage[T1]{fontenc}
\usepackage{url}
\usepackage{booktabs}
\usepackage{amsfonts,amsmath,amssymb,amsthm,mathtools,bm}
\usepackage{nicefrac}
\usepackage{microtype}
\usepackage{xcolor}
\usepackage{tabularx,array,multirow}
\usepackage{enumitem}
\usepackage{thmtools,thm-restate}
\usepackage{algorithm}
\usepackage[noend]{algpseudocode}
\algrenewcommand\algorithmiccomment[1]{\hfill{\color{gray!70!black}\footnotesize #1}}
\usepackage{tikz}
\usetikzlibrary{positioning,arrows.meta,calc,decorations.markings,shapes.geometric,decorations.pathreplacing}
\usepackage{pgfplots}
\pgfplotsset{compat=1.16}
\usepgfplotslibrary{groupplots,fillbetween}
\usepackage{hyperref}
\usepackage[capitalise,noabbrev]{cleveref}

\hypersetup{colorlinks=true,linkcolor=blue!50!black,citecolor=blue!50!black,urlcolor=blue!50!black,
  pdftitle={Say, Echo, Do: Strategic Narratives and Revealed Positioning in Financial Markets},pdfauthor={Ali Atiah Alzahrani}}

\definecolor{say}{HTML}{8E44AD}
\definecolor{echo}{HTML}{C0392B}
\definecolor{do}{HTML}{1F618D}
\definecolor{price}{HTML}{222222}
\definecolor{good}{HTML}{1E8449}

\theoremstyle{plain}
\newtheorem{theorem}{Theorem}
\newtheorem{proposition}[theorem]{Proposition}
\newtheorem{lemma}[theorem]{Lemma}

\theoremstyle{definition}
\newtheorem{definition}[theorem]{Definition}
\newtheorem{assumption}{Assumption}
\theoremstyle{remark}
\newtheorem{remark}[theorem]{Remark}
\crefname{assumption}{Assumption}{Assumptions}
\crefname{definition}{Definition}{Definitions}
\crefname{remark}{Remark}{Remarks}
\setlist{itemsep=1pt,topsep=2pt}

\newcommand{\R}{\mathbb{R}}
\newcommand{\E}{\mathbb{E}}
\newcommand{\Prob}{\mathbb{P}}
\newcommand{\N}{\mathcal{N}}
\newcommand{\F}{\mathcal{F}}
\newcommand{\z}{\bm{z}}
\newcommand{\w}{\bm{w}}
\newcommand{\y}{\bm{y}}
\newcommand{\h}{\bm{h}}
\newcommand{\bphi}{\bm{\varphi}}
\newcommand{\ind}{\mathbf{1}}
\DeclareMathOperator{\sgn}{sgn}
\DeclareMathOperator{\Var}{Var}
\DeclareMathOperator{\Cov}{Cov}
\DeclareMathOperator{\KL}{KL}
\newcommand{\SAY}{\textcolor{say}{\textsf{Say}}}
\newcommand{\ECHO}{\textcolor{echo}{\textsf{Echo}}}
\newcommand{\DO}{\textcolor{do}{\textsf{Do}}}

\title{Say, Echo, Do: Strategic Narratives and\\ Revealed Positioning in Financial Markets}

\author{%
  Ali Atiah Alzahrani\thanks{First draft, September 2026. Corresponding author: \href{mailto:aliatiah100@gmail.com}{aliatiah100@gmail.com}. The views expressed are those of the author and do not reflect the views of any other individual or entity. This material is for research purposes only and does not constitute investment advice.}\\
  {\normalfont Riyadh, Saudi Arabia}
}

\begin{document}
\maketitle

\begin{abstract}
Machine-learning signals built from financial text treat what institutions say, and what the media repeat, as evidence about value. But whoever shapes a narrative may be trading against it. We study markets with three observable voices: institutional statements (\emph{Say}), media repetition (\emph{Echo}) and revealed positioning (\emph{Do}). We ask when words should be followed and when they should be faded. In a linear--quadratic model of an informed institution that speaks and trades before a partly credulous crowd, talking an asset down while buying it is optimal exactly when $\varphi^2<2\lambda k<\varphi$. A distribution-free identity then shows that when the observable Say--Do covariance is negative, words carry negative predictive content and should be faded. For measurement, we derive (i)~an exact factorised posterior over which articles are echoes, combining arrival times with embedding similarity; (ii)~a return-aligned contrastive objective that attains its bound exactly when squared embedding distances are an increasing affine function of squared outcome distances, with the tightest loss-based certificate of which neighbour rankings survive imperfect training; and (iii)~a path-signature statistic for who moved first. In a controlled market with known ground truth, echo sentiment predicts returns with a significantly negative sign in all $29$ simulated markets, the rolling Say--Do correlation flags false-alarm events with an AUC of $0.90$, and return-aligned embeddings organise headlines by consequence rather than topic. We also report where the tools fail.
\end{abstract}

\section{Introduction}\label{sec:intro}

A bank downgrades a stock. Within hours newswires rewrite the note, commentators repeat it and social media amplifies it. The price falls, but by less than the volume of bad news would suggest, and weeks later it is higher than before. The next round of disclosures shows that some institutions were buying throughout. Traders call this a \emph{false alarm}.

Most machine-learning pipelines for financial text would score this episode as strongly negative. Sentiment classifiers, domain encoders and prompted language models map a document to a number that is assumed to move with value \citep{ke2019,araci2019,lopezlira2023}. That assumption fails when text is \emph{strategic}. The author of a narrative may profit from the crowd's reaction to it, and repetition inflates a signal without adding information. The useful question is therefore not only what a text says. It is also whether its author's money agrees, and whether the text is news or an echo of earlier news.

We separate three voices that any market participant can observe (\cref{fig:overview}). \SAY\ is what institutions state: ratings, notes, management tone. \ECHO\ is how the media repeat and amplify it. \DO\ is where capital is revealed to have gone: holdings, insider trades, futures positioning, order flow. Words are cheap and easy to observe, repetition is mostly redundant, and positions are costly, slow to be disclosed and hard to fake.

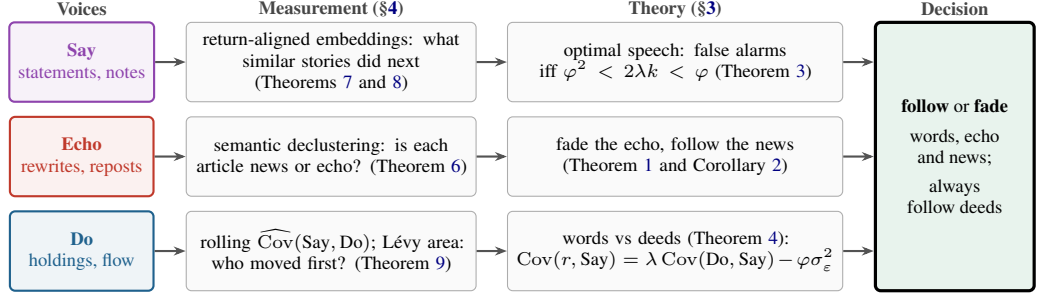
\begin{figure}[t]
\centering
\begin{tikzpicture}[
  voice/.style={draw, rounded corners=2pt, thick, minimum width=1.9cm, minimum height=1.05cm, align=center, font=\scriptsize},
  box/.style={draw=gray!65, rounded corners=2pt, fill=gray!4, minimum height=1.05cm, align=center, font=\scriptsize},
  mbox/.style={box, text width=3.6cm, minimum width=3.8cm},
  tbox/.style={box, text width=4.2cm, minimum width=4.4cm},
  arr/.style={-{Stealth[length=1.6mm]}, thick, gray!65!black}]
\foreach \x/\t in {0.95/Voices, 4.25/Measurement (\S\ref{sec:measure}), 8.8/Theory (\S\ref{sec:theory}), 12.5/Decision}
  \node[font=\scriptsize\bfseries, text=gray!55!black] at (\x,1.95) {\t};
\node[voice, draw=say, fill=say!7, text=say] (say) at (0.95,1.25) {\textbf{Say}\\ statements, notes};
\node[voice, draw=echo, fill=echo!7, text=echo] (echo) at (0.95,0) {\textbf{Echo}\\ rewrites, reposts};
\node[voice, draw=do, fill=do!7, text=do] (do) at (0.95,-1.25) {\textbf{Do}\\ holdings, flow};
\node[mbox] (m1) at (4.25,1.25) {return-aligned embeddings: what\\ similar stories did next (\cref{thm:racl,thm:robustnn})};
\node[mbox] (m2) at (4.25,0) {semantic declustering: is each\\ article news or echo? (\cref{thm:semantic})};
\node[mbox] (m3) at (4.25,-1.25) {rolling $\widehat{\Cov}(\text{Say},\text{Do})$; L\'evy area:\\ who moved first? (\cref{thm:levy})};
\node[tbox] (t1) at (8.8,1.25) {optimal speech: false alarms\\ iff $\varphi^2<2\lambda k<\varphi$ (\cref{thm:speech})};
\node[tbox] (t2) at (8.8,0) {fade the echo, follow the news\\ (\cref{thm:followfade,cor:echo})};
\node[tbox] (t3) at (8.8,-1.25) {words vs deeds (\cref{thm:identity}):\\ $\Cov(r,\text{Say})=\lambda\Cov(\text{Do},\text{Say})-\varphi\sigma_\varepsilon^2$};
\node[draw, rounded corners=2pt, very thick, fill=good!10, minimum width=2.1cm, minimum height=3.55cm, align=center, font=\scriptsize] (dec) at (12.5,0) {\textbf{follow} or \textbf{fade}\\[3pt] words, echo\\ and news;\\[3pt] always\\ follow deeds};
\draw[arr] (say) -- (m1); \draw[arr] (echo) -- (m2); \draw[arr] (do) -- (m3);
\draw[arr] (m1) -- (t1); \draw[arr] (m2) -- (t2); \draw[arr] (m3) -- (t3);
\draw[arr] (t1) -- (t1 -| dec.west); \draw[arr] (t2) -- (dec); \draw[arr] (t3) -- (t3 -| dec.west);
\end{tikzpicture}
\caption{Overview. Three observable voices are turned into measurements, and the theory says which sign each measurement should carry. The central object is the Say--Do covariance: when institutions' words and positions move in opposite directions, their words should be faded.}
\label{fig:overview}
\end{figure}

\paragraph{Contributions.}
\begin{itemize}[leftmargin=1.3em]
  \item \textbf{A theory of strategic narratives} (\cref{sec:theory}). We solve for the optimal message and trade of an informed institution facing a partly credulous crowd. False alarms, in which the institution talks an asset down while buying it, are optimal in an explicit wedge of parameters. A distribution-free identity shows that a negative Say--Do covariance means words should be faded. Media virality selects which kind of manipulation appears, and adaptive trust can cycle and, numerically, become chaotic: repeated re-estimation of the reading rule fails to converge, as in repeated risk minimisation under performativity \citep{perdomo2020}.
  \item \textbf{Measurement with guarantees} (\cref{sec:measure}). An exact posterior over echo parentage that uses both timing and meaning. A return-aligned contrastive objective that attains its lower bound exactly when embedding geometry reproduces outcome geometry, together with the tightest certificate, among those based on the excess loss alone, of which neighbour rankings survive imperfect training. A signature statistic that identifies lead and lag.
  \item \textbf{A controlled evaluation} (\cref{sec:experiments}). A simulated market in which every ground-truth quantity is known, a synthetic headline benchmark, and a report of both successes and failures, including where the theory adds nothing a linear model cannot learn. Code, tests and every experiment are released at \url{https://github.com/AliAtiah/say-echo-do}.
\end{itemize}

\paragraph{What we do not claim.} A ``strategic'' reading means only that words and money moved in opposite directions, in a particular order. It is not evidence of deception. We do not test the theory on market data; our experiments are simulated so that the truth is known.

\section{Related work}\label{sec:related}

\paragraph{Financial text.} Media pessimism predicts price pressure and reversal \citep{tetlock2007}, and prices drift after news but reverse after no-news moves \citep{chan2003}. Supervised text factors \citep{ke2019}, domain encoders \citep{araci2019} and prompted language models \citep{lopezlira2023} predict returns, subject to look-ahead leakage \citep{glasserman2023}. These methods take text at face value. We model the author's incentive to shape it, and use revealed positions to discount it.

\paragraph{Strategic behaviour and performativity.} Cheap talk \citep{crawford1982}, costly talk with credulous receivers \citep{kartik2007} and reputational manipulation \citep{benabou1992} study how informed agents shape beliefs. Closest to us, informed traders can profit by spreading rumours and trading against them \citep{vanbommel2003}, and promotional articles move prices \citep{kogan2023}. Order flow reveals private information \citep{kyle1985}. In machine learning, strategic classification \citep{hardt2016} and performative prediction \citep{perdomo2020} study predictions that change the data they are evaluated on. Our institution games a known reading rule, and our credibility cycles are a closed-form case in which repeated re-estimation of that rule fails to converge.

\paragraph{Point processes, representations and signatures.} We build on the cluster representation of Hawkes processes \citep{hawkes1971,hawkesoakes1974}. Our declustering posterior has the structure of space--time ETAS declustering \citep{zhuang2002}, with embeddings in place of locations and offspring marks that depend on the parent's. Text-aware point processes cluster document streams or infer diffusion networks \citep{du2015,he2015}, and neural Hawkes models learn intensities \citep{mei2017}, but none attributes each article to news or echo. Our embedding objective is the stochastic neighbour embedding loss \citep{hinton2002} with market outcomes as the input space, a link between contrastive learning and SNE that is known \citep{hu2023,damrich2023}. It is close to contrastive learning with continuous labels \citep{oord2018,khosla2020,dufumier2021,zha2023}. Our additions are the characterisation of when the bound is attained, which is elementary, and a tight certificate for imperfect training. L\'evy area is the antisymmetric part of the level-two path signature \citep{lyons1998,chevyrev2016,kidger2019}, and has been used to detect lead--lag between asset returns \citep{bennett2022}.

\section{Theory: three voices and strategic speech}\label{sec:theory}

\subsection{A three-voice market}\label{sec:model}

Value is $v\sim\N(0,\sigma_v^2)$. The market observes $K$ public channels $\y=\h v+\bm\nu$ with Gaussian noise, labelled by voice: statements, independent news and echo. Informed traders submit $x=\beta(v-p)$, noise traders submit $u\sim\N(0,\sigma_u^2)$, and a narrative-sensitive crowd sets the price $p=\bphi^\top\y+\lambda(x+u)$, where $\bphi$ collects the crowd's \emph{narrative multipliers} (\cref{ass:signals,ass:price}, \cref{app:model}). Write $r=v-p$ for the forward return, $\w^\star=\sigma_v^2\Sigma_y^{-1}\h$ for the Bayes weights, and $\kappa=(1+\lambda\beta)^{-1}$.

\begin{restatable}[Follow or fade]{theorem}{thmfollowfade}\label{thm:followfade}
Under \cref{ass:signals,ass:price}, $\E[r\mid\y]=\kappa(\w^\star-\bphi)^\top\y$ and $\Var(r\mid\y)=\kappa^2(\sigma^2_{v|y}+\lambda^2\sigma_u^2)\eqqcolon\varsigma^2$, where $\sigma^2_{v|y}=\Var(v\mid\y)$. The return loads negatively on channel $k$ (fade) iff $\varphi_k>w^\star_k$, and positively (follow) iff $\varphi_k<w^\star_k$.
\end{restatable}

\begin{restatable}[Fade the echo, follow the news]{corollary}{corecho}\label{cor:echo}\label{cor:news}
Under \cref{ass:signals,ass:price}: (i)~if channel $e$ is a garbling of the others, $e=\bm a^\top\y_{-e}+\xi$ with $\xi\sim\N(0,\sigma_\xi^2)$, $\sigma_\xi^2>0$, independent of $(v,\y_{-e},u)$, then $w^\star_e=0$ and the return loads on $e$ with coefficient $-\kappa\varphi_e$, so any crowd weight on echo is reversed; (ii)~if channel $n=v+\nu_n$ with $\nu_n\sim\N(0,\sigma_n^2)$ independent of $(v,\y_{-n},u)$, and $s^2\coloneqq\Var(v\mid\y_{-n})>0$, then $w^\star_n=s^2/(s^2+\sigma_n^2)\in(0,1)$, and returns drift in the direction of $n$ whenever the crowd underweights it, $\varphi_n<w^\star_n$.
\end{restatable}

Two further results, proved in \cref{app:model}, are used in the experiments. When positioning is observed with noise, the optimal forecast depends on words and deeds only through a weighted \emph{Say--Do gap} (\cref{thm:saydo}). And the price move net of the narrative, called \emph{absorption}, predicts continuation if and only if informed flow dominates noise, $\beta\sigma^2_{v|y}>\lambda\sigma_u^2$ (\cref{thm:absorption}).

\subsection{Strategic speech}\label{sec:strategic}

We now let an institution choose what it says. It knows $v$ (any distribution with mean zero and finite variance), publishes a message $m$ that the public sees with noise, $\tilde m=m+\varepsilon$, and trades $x$. The price is $p=\varphi\tilde m+\lambda x$, where $\varphi$ is the crowd's credulity, and misreporting costs $\frac k2(m-v)^2$. The institution maximises $x(v-\varphi m-\lambda x)-\frac k2(m-v)^2$ (\cref{ass:strategic}). Let $a=2\lambda k$.

\begin{restatable}[Optimal speech and its taxonomy]{theorem}{thmspeech}\label{thm:speech}\label{cor:taxonomy}
Under \cref{ass:strategic}, a unique optimum exists iff $a>\varphi^2$ (otherwise the objective is unbounded), and then $m^\ast=\psi v$ and $x^\ast=\chi v$ with $\psi=\frac{a-\varphi}{a-\varphi^2}$ and $\chi=\frac{k(1-\varphi)}{a-\varphi^2}$. For $a>\varphi^2$, $\varphi>0$ and good news ($v>0$) there are four regimes: (i)~\emph{truth} ($\varphi=1$: $\psi=1$, $\chi=0$); (ii)~\emph{shading} ($\varphi<1$, $a>\varphi$: $0<\psi<1$, $\chi>0$); (iii)~\emph{false alarm} ($\varphi^2<a<\varphi<1$: $\psi<0<\chi$, talk it down and buy); and (iv)~\emph{exaggeration} ($\varphi>1$: $\psi>1$, $\chi<0$, talk it up and sell). On the boundary $a=\varphi<1$ the institution says nothing, $\psi=0$.
\end{restatable}

\emph{Sketch.} The objective is a concave quadratic iff its Hessian has determinant $a-\varphi^2>0$. The first-order conditions give $\psi$ and $\chi$, and the regimes follow from the signs of $a-\varphi$ and $1-\varphi$.

\begin{restatable}[Say--Do covariance identity]{theorem}{thmidentity}\label{thm:identity}
Under \cref{ass:strategic}, let the message rule $m(v)$ be arbitrary and the trade optimal given the message. Then $r=\lambda x-\varphi\varepsilon$ pathwise, so for any distribution of $v$,
\[
  \Cov(r,\tilde m)=\lambda\,\Cov(x,\tilde m)-\varphi\,\sigma_\varepsilon^2,\qquad \E[r\mid x]=\lambda x .
\]
In the equilibrium of \cref{thm:speech}, $\Cov(x,\tilde m)=\psi\chi\sigma_v^2<0$ exactly in the false-alarm and exaggeration regimes.
\end{restatable}

\emph{Sketch.} Given its message, the institution's optimal trade solves $v-\varphi m=2\lambda x$; substituting into $r=v-\varphi(m+\varepsilon)-\lambda x$ gives the pathwise identity. \Cref{thm:identity} is the paper's central tool. Both sides are estimable and no distributional assumption is needed. A negative rolling $\widehat{\Cov}(\text{Say},\text{Do})$ says that an institution's words should be faded; following them requires $\lambda\Cov(\text{Do},\text{Say})>\varphi\sigma_\varepsilon^2$. Deeds are always followed.

\paragraph{Virality and adaptive trust.} If each article moves the price by $\varphi_0$ and messages spread through a Hawkes cascade with branching ratio $n$, effective credulity is $\varphi=\varphi_0/(1-n)$. Virality then decides which manipulation appears: false alarms in deep markets ($a<1$) and exaggeration in shallow ones ($a>1$), each inside an explicit window of the echo share (\cref{prop:virality}). If instead the crowd re-estimates its credulity each round as the regression slope of value on (noiseless) messages, then $\varphi_{t+1}=F_a(\varphi_t)$ with $F_a(\varphi)=(a-\varphi^2)/(a-\varphi)$, on the admissible set where $\varphi^2<a$, $\varphi\neq a$ and $\psi(\varphi)\neq0$. Trust then behaves like a predictor that is retrained on the data its own use generates.

\begin{restatable}[Credibility cycles]{theorem}{thmcycles}\label{thm:cycles}
Under the learning rule above, truth ($\varphi=1$) is the unique admissible fixed point for $a>1$, and it is locally stable iff $\lambda k>1$. For $1<a<2$ there is a unique $2$-cycle $\varphi_\pm=\frac12\big(a\pm\sqrt{a(2-a)}\big)$ that alternates between shading and exaggeration, with multiplier $\frac{5a-9}{a-1}$. It is locally stable iff $\frac56<\lambda k<1$. It is born from truth in a flip bifurcation at $a=2$ and loses stability in a second flip at $a=\frac53$.
\end{restatable}

\begin{figure}[t]
\centering
\begin{minipage}[t]{0.49\textwidth}
\centering
\begin{tikzpicture}
\begin{axis}[width=6.6cm, height=4.7cm, xmin=0, xmax=1.6, ymin=0, ymax=2.6,
  xlabel={\scriptsize crowd credulity $\varphi$},
  ylabel={\scriptsize deterrence $a=2\lambda k$},
  tick label style={font=\scriptsize}, axis on top, clip=false,
  xlabel style={yshift=3pt}, ylabel style={yshift=-4pt}]
\addplot[name path=sq, draw=black, thick, domain=0:1.6, samples=60] {x^2};
\addplot[name path=lin, draw=black, thick, dashed, domain=0:1, samples=2] {x};
\addplot[name path=top, draw=none, domain=0:1.6, samples=2] {2.6};
\addplot[name path=bot, draw=none, domain=0:1.6, samples=2] {0};
\addplot[gray!25] fill between[of=sq and bot];
\addplot[echo!30] fill between[of=lin and sq, soft clip={domain=0:1}];
\addplot[good!18] fill between[of=top and lin, soft clip={domain=0:1}];
\addplot[say!20] fill between[of=top and sq, soft clip={domain=1:1.6}];
\draw[do, very thick] (axis cs:1,1) -- (axis cs:1,2.6);
\node[font=\scriptsize\bfseries, text=good!50!black] at (axis cs:0.45,1.8) {shading};
\node[font=\scriptsize\bfseries, text=echo!80!black, align=center] at (axis cs:0.75,0.64) {false\\[-3pt] alarm};
\node[font=\scriptsize\bfseries, text=say!80!black, align=center] at (axis cs:1.33,2.3) {exaggerate};
\node[font=\scriptsize, text=gray!60!black, align=center] at (axis cs:1.32,0.6) {no optimum};
\node[font=\scriptsize, text=do, anchor=west] at (axis cs:1.01,1.15) {truth};
\draw[-{Stealth[length=1.6mm]}, thick, echo] (axis cs:0.15,0.13) -- (axis cs:0.55,0.13) node[right, font=\scriptsize] {virality};
\end{axis}
\node[font=\bfseries\small] at (-0.75,3.55) {a};
\end{tikzpicture}
\end{minipage}\hfill
\begin{minipage}[t]{0.49\textwidth}
\centering
\begin{tikzpicture}
\begin{axis}[width=6.6cm, height=4.7cm, xmin=1.55, xmax=2.4, ymin=0, ymax=1.35,
  xlabel={\scriptsize deterrence $a=2\lambda k$},
  ylabel={\scriptsize long-run credulity $\varphi_t$},
  tick label style={font=\scriptsize}, clip=false,
  xlabel style={yshift=3pt}, ylabel style={yshift=-4pt}]
\fill[echo!8] (axis cs:1.55,0) rectangle (axis cs:1.6667,1.35);
\fill[say!8] (axis cs:1.6667,0) rectangle (axis cs:2,1.35);
\fill[good!8] (axis cs:2,0) rectangle (axis cs:2.4,1.35);
\addplot[only marks, mark=*, mark size=0.3pt, black] coordinates {(1.565,0.014) (1.565,0.039) (1.565,0.053) (1.565,0.054) (1.565,0.209) (1.565,0.216) (1.565,0.238) (1.565,0.282) (1.565,0.435) (1.565,0.552) (1.565,0.637) (1.565,0.678) (1.565,0.692) (1.565,0.772) (1.565,0.775) (1.565,0.798) (1.565,0.815) (1.565,0.844) (1.565,0.934) (1.565,0.935) (1.565,0.941) (1.565,0.946) (1.565,0.954) (1.565,0.983) (1.565,1.009) (1.565,1.024) (1.565,1.028) (1.565,1.03) (1.565,1.033) (1.565,1.034) (1.565,1.072) (1.565,1.082) (1.565,1.089) (1.565,1.096) (1.565,1.097) (1.565,1.122) (1.565,1.126) (1.565,1.137) (1.565,1.158) (1.565,1.183) (1.565,1.21) (1.565,1.217) (1.565,1.221) (1.565,1.222) (1.565,1.244) (1.565,1.246) (1.565,1.249) (1.57,0.049) (1.57,0.05) (1.57,0.06) (1.57,0.096) (1.57,0.156) (1.57,0.176) (1.57,0.178) (1.57,0.184) (1.57,0.268) (1.57,0.294) (1.57,0.535) (1.57,0.571) (1.57,0.587) (1.57,0.662) (1.57,0.74) (1.57,0.75) (1.57,0.751) (1.57,0.753) (1.57,0.786) (1.57,0.796) (1.57,0.878) (1.57,0.903) (1.57,0.927) (1.57,0.94) (1.57,1.032) (1.57,1.037) (1.57,1.048) (1.57,1.059) (1.57,1.09) (1.57,1.093) (1.57,1.104) (1.57,1.105) (1.57,1.108) (1.57,1.131) (1.57,1.151) (1.57,1.154) (1.57,1.163) (1.57,1.21) (1.57,1.214) (1.57,1.228) (1.57,1.229) (1.57,1.232) (1.57,1.24) (1.57,1.245) (1.57,1.246) (1.57,1.247) (1.575,0.063) (1.575,0.069) (1.575,0.124) (1.575,0.133) (1.575,0.18) (1.575,0.185) (1.575,0.199) (1.575,0.204) (1.575,0.255) (1.575,0.33) (1.575,0.403) (1.575,0.444) (1.575,0.473) (1.575,0.618) (1.575,0.709) (1.575,0.715) (1.575,0.719) (1.575,0.741) (1.575,0.744) (1.575,0.75) (1.575,0.826) (1.575,0.839) (1.575,0.916) (1.575,0.924) (1.575,1.039) (1.575,1.043) (1.575,1.075) (1.575,1.08) (1.575,1.106) (1.575,1.108) (1.575,1.109) (1.575,1.116) (1.575,1.117) (1.575,1.119) (1.575,1.144) (1.575,1.178) (1.575,1.184) (1.575,1.192) (1.575,1.205) (1.575,1.218) (1.575,1.226) (1.575,1.227) (1.575,1.229) (1.575,1.23) (1.575,1.237) (1.575,1.238) (1.575,1.246) (1.575,1.247) (1.58,0.087) (1.58,0.09) (1.58,0.628) (1.58,0.644) (1.58,0.889) (1.58,0.894) (1.58,1.053) (1.58,1.055) (1.58,1.138) (1.58,1.142) (1.58,1.245) (1.585,0.109) (1.585,0.139) (1.585,0.14) (1.585,0.141) (1.585,0.223) (1.585,0.243) (1.585,0.256) (1.585,0.335) (1.585,0.342) (1.585,0.346) (1.585,0.409) (1.585,0.412) (1.585,0.416) (1.585,0.462) (1.585,0.471) (1.585,0.484) (1.585,0.553) (1.585,0.629) (1.585,0.652) (1.585,0.687) (1.585,0.819) (1.585,0.82) (1.585,0.821) (1.585,0.865) (1.585,1.066) (1.585,1.083) (1.585,1.084) (1.585,1.127) (1.585,1.137) (1.585,1.143) (1.585,1.162) (1.585,1.178) (1.585,1.181) (1.585,1.183) (1.585,1.192) (1.585,1.193) (1.585,1.194) (1.585,1.206) (1.585,1.208) (1.585,1.221) (1.585,1.224) (1.585,1.227) (1.585,1.239) (1.585,1.24) (1.585,1.244) (1.59,0.13) (1.59,0.131) (1.59,0.132) (1.59,0.133) (1.59,0.252) (1.59,0.254) (1.59,0.264) (1.59,0.268) (1.59,0.272) (1.59,0.462) (1.59,0.465) (1.59,0.467) (1.59,0.474) (1.59,0.476) (1.59,0.614) (1.59,0.621) (1.59,0.638) (1.59,0.642) (1.59,0.832) (1.59,0.833) (1.59,0.834) (1.59,0.836) (1.59,0.837) (1.59,1.078) (1.59,1.079) (1.59,1.141) (1.59,1.142) (1.59,1.147) (1.59,1.148) (1.59,1.181) (1.59,1.182) (1.59,1.183) (1.59,1.184) (1.59,1.22) (1.59,1.221) (1.59,1.222) (1.59,1.223) (1.59,1.224) (1.59,1.242) (1.59,1.243) (1.595,0.153) (1.595,0.167) (1.595,0.188) (1.595,0.225) (1.595,0.278) (1.595,0.297) (1.595,0.332) (1.595,0.349) (1.595,0.361) (1.595,0.374) (1.595,0.384) (1.595,0.388) (1.595,0.402) (1.595,0.408) (1.595,0.416) (1.595,0.434) (1.595,0.456) (1.595,0.508) (1.595,0.571) (1.595,0.604) (1.595,0.693) (1.595,0.752) (1.595,0.785) (1.595,0.806) (1.595,1.09) (1.595,1.097) (1.595,1.109) (1.595,1.127) (1.595,1.152) (1.595,1.161) (1.595,1.176) (1.595,1.187) (1.595,1.192) (1.595,1.195) (1.595,1.197) (1.595,1.198) (1.595,1.201) (1.595,1.202) (1.595,1.204) (1.595,1.206) (1.595,1.208) (1.595,1.212) (1.595,1.218) (1.595,1.221) (1.595,1.23) (1.595,1.236) (1.595,1.239) (1.595,1.241) (1.6,0.171) (1.6,0.173) (1.6,0.184) (1.6,0.203) (1.6,0.246) (1.6,0.26) (1.6,0.269) (1.6,0.289) (1.6,0.315) (1.6,0.339) (1.6,0.35) (1.6,0.358) (1.6,0.472) (1.6,0.486) (1.6,0.505) (1.6,0.546) (1.6,0.591) (1.6,0.625) (1.6,0.64) (1.6,0.664) (1.6,0.733) (1.6,0.763) (1.6,0.779) (1.6,0.782) (1.6,1.099) (1.6,1.1) (1.6,1.106) (1.6,1.116) (1.6,1.137) (1.6,1.144) (1.6,1.148) (1.6,1.157) (1.6,1.168) (1.6,1.178) (1.6,1.182) (1.6,1.185) (1.6,1.208) (1.6,1.21) (1.6,1.212) (1.6,1.216) (1.6,1.221) (1.6,1.224) (1.6,1.226) (1.6,1.228) (1.6,1.235) (1.6,1.238) (1.6,1.24) (1.605,0.191) (1.605,0.195) (1.605,0.2) (1.605,0.205) (1.605,0.21) (1.605,0.217) (1.605,0.253) (1.605,0.267) (1.605,0.281) (1.605,0.293) (1.605,0.306) (1.605,0.316) (1.605,0.552) (1.605,0.57) (1.605,0.59) (1.605,0.612) (1.605,0.635) (1.605,0.658) (1.605,0.716) (1.605,0.726) (1.605,0.735) (1.605,0.742) (1.605,0.749) (1.605,0.755) (1.605,1.109) (1.605,1.111) (1.605,1.114) (1.605,1.116) (1.605,1.119) (1.605,1.122) (1.605,1.14) (1.605,1.146) (1.605,1.152) (1.605,1.158) (1.605,1.163) (1.605,1.168) (1.605,1.218) (1.605,1.22) (1.605,1.222) (1.605,1.224) (1.605,1.226) (1.605,1.229) (1.605,1.235) (1.605,1.236) (1.605,1.237) (1.605,1.238) (1.605,1.239) (1.61,0.21) (1.61,0.291) (1.61,0.602) (1.61,0.731) (1.61,1.118) (1.61,1.156) (1.61,1.224) (1.61,1.238) (1.615,0.226) (1.615,0.279) (1.615,0.627) (1.615,0.71) (1.615,1.126) (1.615,1.151) (1.615,1.228) (1.615,1.237) (1.62,0.257) (1.62,0.667) (1.62,1.14) (1.62,1.233) (1.625,0.268) (1.625,0.656) (1.625,1.144) (1.625,1.233) (1.63,0.279) (1.63,0.644) (1.63,1.149) (1.63,1.233) (1.635,0.291) (1.635,0.632) (1.635,1.154) (1.635,1.232) (1.64,0.305) (1.64,0.618) (1.64,1.159) (1.64,1.231) (1.645,0.32) (1.645,0.603) (1.645,1.164) (1.645,1.23) (1.65,0.337) (1.65,0.586) (1.65,1.17) (1.65,1.228) (1.655,0.356) (1.655,0.565) (1.655,1.177) (1.655,1.226) (1.66,0.382) (1.66,0.54) (1.66,1.185) (1.66,1.222) (1.665,0.421) (1.665,0.5) (1.665,1.196) (1.665,1.215) (1.67,0.464) (1.67,1.206) (1.675,0.469) (1.675,1.206) (1.68,0.473) (1.68,1.207) (1.685,0.478) (1.685,1.207) (1.69,0.483) (1.69,1.207) (1.695,0.488) (1.695,1.207) (1.7,0.493) (1.7,1.207) (1.705,0.498) (1.705,1.207) (1.71,0.503) (1.71,1.207) (1.715,0.508) (1.715,1.207) (1.72,0.513) (1.72,1.207) (1.725,0.518) (1.725,1.207) (1.73,0.523) (1.73,1.207) (1.735,0.528) (1.735,1.207) (1.74,0.534) (1.74,1.206) (1.745,0.539) (1.745,1.206) (1.75,0.544) (1.75,1.206) (1.755,0.55) (1.755,1.205) (1.76,0.555) (1.76,1.205) (1.765,0.56) (1.765,1.205) (1.77,0.566) (1.77,1.204) (1.775,0.572) (1.775,1.203) (1.78,0.577) (1.78,1.203) (1.785,0.583) (1.785,1.202) (1.79,0.588) (1.79,1.202) (1.795,0.594) (1.795,1.201) (1.8,0.6) (1.8,1.2) (1.805,0.606) (1.805,1.199) (1.81,0.612) (1.81,1.198) (1.815,0.618) (1.815,1.197) (1.82,0.624) (1.82,1.196) (1.825,0.63) (1.825,1.195) (1.83,0.636) (1.83,1.194) (1.835,0.642) (1.835,1.193) (1.84,0.649) (1.84,1.191) (1.845,0.655) (1.845,1.19) (1.85,0.662) (1.85,1.188) (1.855,0.668) (1.855,1.187) (1.86,0.675) (1.86,1.185) (1.865,0.682) (1.865,1.183) (1.87,0.688) (1.87,1.182) (1.875,0.695) (1.875,1.18) (1.88,0.703) (1.88,1.177) (1.885,0.71) (1.885,1.175) (1.89,0.717) (1.89,1.173) (1.895,0.724) (1.895,1.171) (1.9,0.732) (1.9,1.168) (1.905,0.74) (1.905,1.165) (1.91,0.748) (1.91,1.162) (1.915,0.756) (1.915,1.159) (1.92,0.764) (1.92,1.156) (1.925,0.773) (1.925,1.152) (1.93,0.781) (1.93,1.149) (1.935,0.79) (1.935,1.145) (1.94,0.799) (1.94,1.141) (1.945,0.809) (1.945,1.136) (1.95,0.819) (1.95,1.131) (1.955,0.829) (1.955,1.126) (1.96,0.84) (1.96,1.12) (1.965,0.851) (1.965,1.114) (1.97,0.863) (1.97,1.107) (1.975,0.876) (1.975,1.099) (1.98,0.891) (1.98,1.089) (1.985,0.906) (1.985,1.079) (1.99,0.924) (1.99,1.066) (1.995,0.948) (1.995,1.047) (2.0,0.991) (2.0,1.009) (2.005,1.0) (2.01,1.0) (2.015,1.0) (2.02,1.0) (2.025,1.0) (2.03,1.0) (2.035,1.0) (2.04,1.0) (2.045,1.0) (2.05,1.0) (2.055,1.0) (2.06,1.0) (2.065,1.0) (2.07,1.0) (2.075,1.0) (2.08,1.0) (2.085,1.0) (2.09,1.0) (2.095,1.0) (2.1,1.0) (2.105,1.0) (2.11,1.0) (2.115,1.0) (2.12,1.0) (2.125,1.0) (2.13,1.0) (2.135,1.0) (2.14,1.0) (2.145,1.0) (2.15,1.0) (2.155,1.0) (2.16,1.0) (2.165,1.0) (2.17,1.0) (2.175,1.0) (2.18,1.0) (2.185,1.0) (2.19,1.0) (2.195,1.0) (2.2,1.0) (2.205,1.0) (2.21,1.0) (2.215,1.0) (2.22,1.0) (2.225,1.0) (2.23,1.0) (2.235,1.0) (2.24,1.0) (2.245,1.0) (2.25,1.0) (2.255,1.0) (2.26,1.0) (2.265,1.0) (2.27,1.0) (2.275,1.0) (2.28,1.0) (2.285,1.0) (2.29,1.0) (2.295,1.0) (2.3,1.0) (2.305,1.0) (2.31,1.0) (2.315,1.0) (2.32,1.0) (2.325,1.0) (2.33,1.0) (2.335,1.0) (2.34,1.0) (2.345,1.0) (2.35,1.0) (2.355,1.0) (2.36,1.0) (2.365,1.0) (2.37,1.0) (2.375,1.0) (2.38,1.0) (2.385,1.0) (2.39,1.0) (2.395,1.0) (2.4,1.0)};
\addplot[do, thick, dashed, domain=1.6667:2, samples=60] {(x+sqrt(x*(2-x)))/2};
\addplot[do, thick, dashed, domain=1.6667:2, samples=60] {(x-sqrt(x*(2-x)))/2};
\addplot[gray, thin, dotted] coordinates {(1.55,1) (2.4,1)};
\node[font=\scriptsize\bfseries, text=good!50!black, anchor=south] at (axis cs:2.2,1.36) {truth};
\node[font=\scriptsize\bfseries, text=say!80!black, anchor=south] at (axis cs:1.83,1.36) {2-cycle};
\node[font=\scriptsize\bfseries, text=echo!80!black, anchor=south] at (axis cs:1.608,1.36) {chaos};
\end{axis}
\node[font=\bfseries\small] at (-0.75,3.55) {b};
\end{tikzpicture}
\end{minipage}
\caption{\textbf{a}: regime map of \cref{thm:speech}. False alarms occupy the thin wedge $\varphi^2<a<\varphi$, and virality moves a market to the right. \textbf{b}: bifurcation diagram of adaptive trust (\cref{thm:cycles}). Dots show the numerically computed long-run attractor, and dashed curves show the closed-form $2$-cycle. Trust converges to truth only when $a>2$; below $a=\frac53$ period-doubling continues, and simulations suggest chaos.}
\label{fig:theory}
\end{figure}
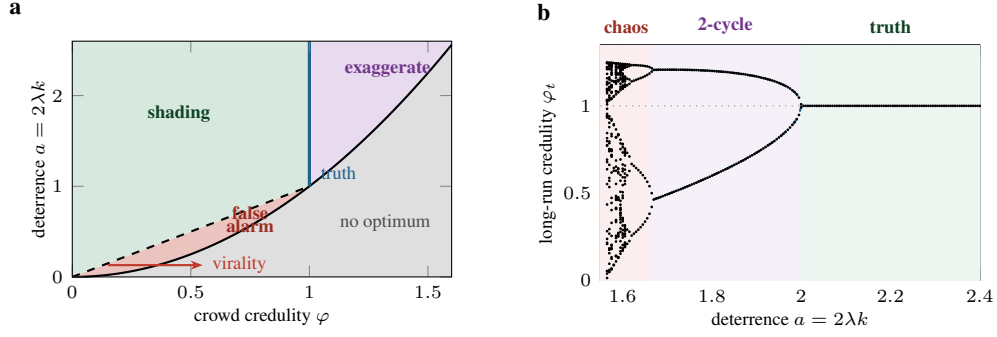

\section{Measuring the voices}\label{sec:measure}

\subsection{Echo: semantic declustering}\label{sec:echo}\label{sec:semantic}

Articles about an asset arrive at times $\tau_1<\dots<\tau_N$ with embeddings $\z_j\in S^{D-1}$. We model them as a Hawkes process with baseline $\mu$ and kernel $g$. In its cluster representation, each article is either an immigrant (\emph{news}) or an offspring of an earlier article (\emph{echo}) \citep{hawkesoakes1974}. We let an echo's embedding depend on its parent's, $\z\sim f(\cdot\mid\z_l)$, for example von Mises--Fisher with concentration $\varkappa$, while news has density $f_0$ (\cref{ass:marked}). Let $\pi(j)$ denote the parent of article $j$, with $\pi(j)=0$ for news.

\begin{restatable}[Semantic declustering]{theorem}{thmsemantic}\label{thm:semantic}
Under \cref{ass:marked}, given all times and embeddings, the parent labels are independent, with
\[
  \Prob(\pi(j)=0\mid\cdot)=\frac{\mu f_0(\z_j)}{\Lambda_j},\qquad
  \Prob(\pi(j)=l\mid\cdot)=\frac{g(\tau_j-\tau_l)\,f(\z_j\mid\z_l)}{\Lambda_j},
\]
where $\Lambda_j=\mu f_0(\z_j)+\sum_{l<j}g(\tau_j-\tau_l)f(\z_j\mid\z_l)$. Consequently, writing $p^{\text{sem}}_j=\mu f_0(\z_j)/\Lambda_j$: (i)~if each sentiment $s_j$ is a function of $\z_j$ (for example $\langle\w_s,\z_j\rangle$), then $s^{\text{new}}=\sum_jp^{\text{sem}}_js_j$ is the MMSE estimate of total news sentiment; (ii)~$\E\,H(\pi(j)\mid\mathcal{T},\z)\le\E\,H(\pi(j)\mid\mathcal{T})$, so meaning never increases the expected posterior entropy of parentage.
\end{restatable}

\emph{Sketch.} In the cluster representation, the density of a labelled configuration is a product over articles of the parent stream's intensity times the mark density, so the posterior over labellings factorises. The posterior uses only each article's past, so the split is causal. It costs $O(N)$ per article, and $\varkappa$ is fitted by maximum marked likelihood. (In our simulator echoes copy their parent's sentiment, so condition (i) holds only approximately.)

\subsection{Consequence: return-aligned embeddings}\label{sec:embeddings}

For documents with embeddings $\z_j$ and post-event outcome vectors $\bm\rho_j\in\R^M$ (returns at several horizons, volume, implied volatility), let $S_{jk}=\langle\z_j,\z_k\rangle$ and $D_{jk}=\|\bm\rho_j-\bm\rho_k\|^2$. Define soft targets and model probabilities over neighbours $k\neq j$ and a soft-InfoNCE loss:
\[
  q_{jk}\propto e^{-D_{jk}/2b^2},\qquad p_{jk}(Z)\propto e^{S_{jk}/\tau},\qquad \mathcal{L}(Z)=-\sum_j\sum_{k\neq j}q_{jk}\log p_{jk}(Z).
\]
This is the loss of stochastic neighbour embedding \citep{hinton2002}, with market outcomes playing the role of the input space.

\begin{restatable}[When the bound is attained]{theorem}{thmracl}\label{thm:racl}
$\mathcal{L}(Z)\ge\sum_jH(q_{j\cdot})$, with equality iff $S_{jk}=c-\frac{\tau}{2b^2}D_{jk}$ for a constant $c$ and all $j\neq k$. At equality $\|\z_j-\z_k\|^2=2(1-c)+\frac{\tau}{b^2}\|\bm\rho_j-\bm\rho_k\|^2$, an increasing affine function of squared outcome distance, so nearest neighbours in embedding space are nearest neighbours in outcome space. (Equality requires this geometry to be realisable on $S^{D-1}$.)
\end{restatable}

\emph{Sketch.} Per anchor, the loss equals $H(q_{j\cdot})+\KL(q_{j\cdot}\|p_{j\cdot})$. Equality forces the logits to differ from $-D_{jk}/2b^2$ by an anchor constant, and the symmetry of $S$ and $D$ makes the constants equal. In practice training never reaches the bound, but the leftover loss can still certify some rankings.

\begin{restatable}[Neighbour certificate]{theorem}{thmrobustnn}\label{thm:robustnn}
For any embedding $Z$, let $\epsilon_j\coloneqq\KL(q_{j\cdot}\|p_{j\cdot}(Z))$ be the excess loss at anchor $j$. For $k,l\neq j$, the embedding ranks $k$ closer to $j$ than $l$ if (i)~$q_{jk}-q_{jl}>\sqrt{2\epsilon_j}$ (Pinsker), or if (ii)~$q_{jk}>q_{jl}$ and
\[
  \Delta_{jkl}\coloneqq q_{jk}\log\frac{2q_{jk}}{q_{jk}+q_{jl}}+q_{jl}\log\frac{2q_{jl}}{q_{jk}+q_{jl}}>\epsilon_j .
\]
Condition (i) implies (ii), and (ii) is the weakest condition valid for every distribution within KL distance $\epsilon_j$ of $q_{j\cdot}$: if $\epsilon_j\ge\Delta_{jkl}$, one such distribution ties $k$ and $l$.
\end{restatable}

Because targets are of order $1/N$, form (ii) tolerates an excess loss at least $1/(q_{jk}+q_{jl})\approx N/2$ times larger than form (i). The proof is a data-processing argument on a three-cell coarse-graining (\cref{app:measure}).

\subsection{Who moved first: L\'evy area}\label{sec:leadlag}

For paths $X$ (positioning) and $Y$ (headline gloom), the L\'evy area $\mathcal{A}_T=\frac12\int_0^T(X_u-X_0)\,dY_u-(Y_u-Y_0)\,dX_u$ is the antisymmetric level-two signature term.

\begin{restatable}[Expected area identifies lead--lag]{theorem}{thmlevy}\label{thm:levy}
Let $(X,Y)$ be zero-mean and jointly weakly stationary with $C^1$ (mean-square) paths and cross-covariance $C(h)=\E[X_sY_{s+h}]\in C^1$. Then $\E[\mathcal{A}_T]=T\,C'(0)-\frac12\big(C(T)-C(-T)\big)$. If moreover $Y_t=\gamma X_{t-\ell}+Z_t$ with $\ell\neq0$, $Z$ stationary and independent of $X$, and $X$ has an even $C^1$ autocovariance $R$ with $R'(h)<0$ for $h>0$, then the long-run area rate has sign $\sgn(\gamma\ell)$. If $\gamma>0$ (gloom follows positioning), a positive area means money moved first.
\end{restatable}

\subsection{From signals to a decision}\label{sec:synthesis}

Each measurement enters a composite score with the sign a result assigns to it:
\begin{equation}\label{eq:composite}
  D=\omega_1s^{\text{new}}-\omega_2s^{\text{echo}}+\omega_3\,\text{Do}+\widehat{\sgn}\,\omega_4\,\text{Say}+\omega_5A+\omega_6\mathcal{A}^{\pm},
\end{equation}
where $\widehat{\sgn}$ is the sign of the rolling Say--Do covariance (\cref{thm:identity}), $A$ is absorption and $\mathcal{A}^{\pm}$ is the L\'evy area signed by the direction of the position change. \Cref{alg:score} computes every term causally. The weights should be fitted on past data, and a fitted weight whose sign contradicts the theory is a diagnostic that an assumption fails. \Cref{app:inference} adds a causal regime filter, a detection bound, adaptive conformal intervals and Kelly sizing under estimation risk.

\begin{algorithm}[t]
\caption{Say--Echo--Do scoring of one event (all inputs available at decision time)}
\label{alg:score}
\small
\begin{algorithmic}[1]
\Require articles $(\tau_j,\z_j,s_j)_{j\le N}$ (the statement is $j=1$); statement Say; disclosed positioning Do; paths of positioning $X$ and tone; the firm's past $(\text{Say},\text{Do})$; Hawkes $(\mu,g)$ and $\varkappa$ fitted on past events; window $W$, threshold $c$
\For{$j=1,\dots,N$} \Comment{semantic declustering, \cref{thm:semantic}}
  \State $\Lambda_j\gets\mu f_0(\z_j)+\sum_{l<j}g(\tau_j-\tau_l)f(\z_j\mid\z_l)$;\quad $p_j\gets\mu f_0(\z_j)/\Lambda_j$
\EndFor
\State $s^{\text{new}}\gets\sum_{j>1}p_js_j$;\quad $s^{\text{echo}}\gets\sum_j(1-p_j)s_j$ \Comment{\cref{cor:echo}}
\State $A\gets$ price reaction $-$ narrative-implied move (fitted on past events, or analogues in the return-aligned space)
\State $\hat c\gets\operatorname{corr}(\text{Say},\text{Do})$ over the last $W$ events;\quad $\widehat{\sgn}\gets-1$ if $\hat c<-c$ else $+1$ \Comment{\cref{thm:identity}}
\State $\mathcal{A}^{\pm}\gets\sgn(X_T-X_0)\cdot\mathcal{A}_T(X,-\text{tone})$ \Comment{\cref{thm:levy}}
\State \Return $D$ from \eqref{eq:composite}, with weights fitted on past events
\end{algorithmic}
\end{algorithm}

\section{Experiments}\label{sec:experiments}

We test the tools where the truth is known and every quantity must be estimated from what an econometrician would see. We ask five questions. \textbf{Q1}: can echo be separated from news, and do the two carry opposite signs? \textbf{Q2}: does the Say--Do covariance find the institutions whose words should be faded? \textbf{Q3}: do absorption and lead--lag carry their predicted signs? \textbf{Q4}: do theory-based features improve forecasts? \textbf{Q5}: do return-aligned embeddings organise text by consequence, and is the certificate informative?

\paragraph{Simulated market.} There are $300$ firms with $60$ events each. Each firm hosts one institution, drawn with equal probability as non-strategic (it reports value and does not trade) or strategic in the shading, false-alarm or exaggeration regime of \cref{thm:speech}. At each event, value is $v\sim\N(0,1)$, the institution trades $\chi v$ over the ten days before it speaks, and then it publishes $\text{Say}=\psi v+\text{noise}$. The media publish a Hawkes cascade in which the statement and independent news items are immigrants. Echoes copy their parent's sentiment and, through a von Mises--Fisher perturbation ($\varkappa=30$, $D=16$), its embedding. The branching ratio is set so that effective credulity is $\varphi$. The crowd prices every article, the forward return is $v-p$ plus heavy late-arriving noise, and positions are disclosed with error. Methods see article times, embeddings and sentiment, statements, noisy positions, price reactions and daily paths. They never see value, regimes, parentage or parameters. The one exception is a supervised false-alarm detector trained on labelled training events.

\paragraph{Protocol.} For each firm, the first $30$ events are for training and the last $30$ for testing. Hawkes parameters, $\varkappa$, the narrative-implied move, and all model weights and standardisations use training events only. We report the cross-sectional information coefficient (IC) per test period with its $t$-statistic, and firm-clustered regressions for the theory tests. The main market uses $5$ seeds, robustness variants $3$, and sample-size variants $2$: $29$ distinct markets in all. The noise scale was calibrated once so that a linear model on the raw voices reaches an IC of about $0.17$. Forecasting with the true mispricing $v-p$ then reaches only $0.181$, so Q4 has little headroom. Full details are in \cref{app:experiments}.

\begin{table}[t]
\caption{Theory tests on test events of the main market (five seeds). $t$: firm-clustered, averaged over seeds; last column: seeds with $|t|>2$ and the predicted sign. Controls: Say and Do for the first two rows, Do for the third.}
\label{tab:mainTests}
\centering
\small
\begin{tabular}{@{}lllrc@{}}
\toprule
Regressor & Source & Prediction & $t$ & As predicted \\
\midrule
Echo sentiment $s^{\text{echo}}$ & \cref{cor:echo} & $-$ (fade) & $-8.0$ & 5/5 \\
News sentiment $s^{\text{new}}$ & \cref{cor:echo} & $+$ (follow) & $+2.1$ & 3/5 \\
Say, where rolling $\operatorname{corr}(\text{Say},\text{Do})<-0.2$ & \cref{thm:identity} & $-$ (fade) & $-3.3$ & 4/5 \\
Absorption, institution trades & \cref{thm:absorption} & either & $-11.1$ & -- \\
Absorption, institution does not trade & \cref{thm:absorption} & $-$ & $-3.6$ & 4/5 \\
Signed L\'evy area & \cref{thm:levy} & $+$ & $+1.0$ & 1/5 \\
\bottomrule
\end{tabular}
\end{table}

\paragraph{Q1: echo is separable, and fading it pays.} Semantic declustering recovers the echo concentration ($\hat\varkappa=30.27\pm0.07$, true $30$) and the echo share ($0.422$, true $0.427$). Timing alone gives $0.368$ against $0.430$ on the same markets, because it mistakes echoes for news. Echo sentiment predicts returns with a significantly negative sign in all $29$ markets (\cref{tab:mainTests}). Meaning sharpens the test: on identical markets its $t$-statistic weakens from $-7.5$ to $-6.0$ when timing alone is used. News carries the predicted positive sign in every market, but it is significant in only $21$ of $29$.

\paragraph{Q2: the Say--Do covariance finds the manipulators.} The rolling Say--Do correlation alone separates false-alarm events from all others with an AUC of $0.895$. A supervised detector that adds four features reaches $0.934\pm0.010$ (\cref{fig:experiments}a). Used as the switch in \eqref{eq:composite}, the correlation labels a firm's words correctly (fade or follow) $85\%$ of the time from $5$ past events and $96\%$ from $40$ (\cref{fig:experiments}b). In the flagged group, statements predict returns with the sign reversed. Two caveats matter for real data. Controlling for echo sentiment removes this effect ($t=+0.2$; $-0.3$ with news also controlled), because a statement's price impact travels through its echoes. And Do's own coefficient is imprecise, because Say and Do are nearly collinear within strategic regimes and because the institution's price impact is small next to later information.

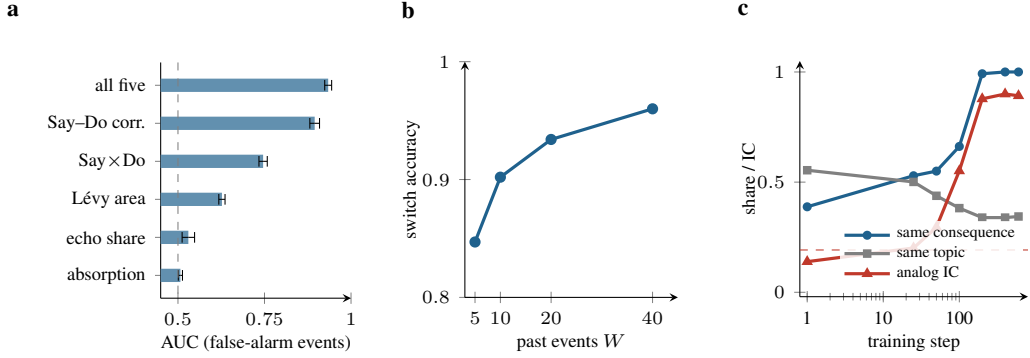
\begin{figure}[t]
\centering
\begin{minipage}[t]{0.35\textwidth}
\centering
\begin{tikzpicture}
\begin{axis}[xbar, width=4.1cm, height=4.7cm, xmin=0.45, xmax=1.0, bar width=5pt,
  ytick={1,...,6}, y dir=reverse, ymin=0.4, ymax=6.6,
  yticklabels={all five, Say--Do corr., Say$\times$Do, L\'evy area, echo share, absorption},
  xtick={0.5,0.75,1.0}, xlabel={\scriptsize AUC (false-alarm events)},
  tick label style={font=\scriptsize}, axis lines=left, y axis line style={-},
  xticklabel style={/pgf/number format/fixed, /pgf/number format/precision=2},
  xlabel style={yshift=3pt}]
\addplot[fill=do!70, draw=none, error bars/.cd, x dir=both, x explicit] coordinates {(0.934,1) +- (0.010,0) (0.895,2) +- (0.014,0) (0.746,3) +- (0.012,0) (0.627,4) +- (0.009,0) (0.530,5) +- (0.018,0) (0.507,6) +- (0.006,0)};
\draw[gray, dashed] (axis cs:0.5,0.4) -- (axis cs:0.5,6.6);
\end{axis}
\node[font=\bfseries\small] at (-1.95,3.8) {a};
\end{tikzpicture}
\end{minipage}\hfill
\begin{minipage}[t]{0.29\textwidth}
\centering
\begin{tikzpicture}
\begin{axis}[width=4.4cm, height=4.7cm, xmin=3, xmax=45, ymin=0.8, ymax=1.0,
  xtick={5,10,20,40}, ytick={0.8,0.9,1.0},
  yticklabel style={/pgf/number format/fixed, /pgf/number format/precision=1},
  xlabel={\scriptsize past events $W$}, ylabel={\scriptsize switch accuracy},
  tick label style={font=\scriptsize}, axis lines=left, clip=false,
  xlabel style={yshift=3pt}, ylabel style={yshift=-6pt}]
\addplot[do, very thick, mark=*, mark size=1.5pt] coordinates {(5,0.847) (10,0.902) (20,0.934) (40,0.960)};
\end{axis}
\node[font=\bfseries\small] at (-0.75,3.8) {b};
\end{tikzpicture}
\end{minipage}\hfill
\begin{minipage}[t]{0.33\textwidth}
\centering
\begin{tikzpicture}
\begin{semilogxaxis}[width=4.6cm, height=4.7cm, xmin=0.8, xmax=800, ymin=-0.02, ymax=1.05,
  xtick={1,10,100}, xticklabels={1,10,100},
  xlabel={\scriptsize training step}, ylabel={\scriptsize share / IC},
  tick label style={font=\scriptsize}, axis lines=left,
  xlabel style={yshift=3pt}, ylabel style={yshift=-6pt},
  legend style={font=\tiny, draw=none, fill=white, fill opacity=0.85, text opacity=1, at={(0.98,0.03)}, anchor=south east, inner sep=1pt, row sep=-2pt}, legend cell align=left]
\addplot[do, very thick, mark=*, mark size=1.1pt] coordinates {(1,0.388) (25,0.529) (50,0.550) (100,0.662) (200,0.992) (400,1.000) (600,1.000)};
\addlegendentry{same consequence}
\addplot[gray, very thick, mark=square*, mark size=1.0pt] coordinates {(1,0.554) (25,0.501) (50,0.438) (100,0.382) (200,0.339) (400,0.339) (600,0.344)};
\addlegendentry{same topic}
\addplot[echo, very thick, mark=triangle*, mark size=1.3pt] coordinates {(1,0.139) (25,0.201) (50,0.295) (100,0.551) (200,0.878) (400,0.900) (600,0.892)};
\addlegendentry{analog IC}
\addplot[echo, dashed, domain=0.8:800, samples=2, forget plot] {0.192};
\end{semilogxaxis}
\node[font=\bfseries\small] at (-0.75,3.8) {c};
\end{tikzpicture}
\end{minipage}
\caption{\textbf{a}: AUC for flagging false-alarm events, for each feature alone (reported as $\max(\text{AUC},1-\text{AUC})$) and for a detector on all five trained with true labels (mean $\pm$ s.d., five seeds; dashed line: chance). \textbf{b}: accuracy of the Say--Do switch in labelling words to fade. \textbf{c}: return-aligned training moves headlines' ten nearest neighbours from sharing a topic to sharing a consequence; the dashed line is the analog IC of TF-IDF neighbours (three seeds).}
\label{fig:experiments}
\end{figure}

\paragraph{Q3: absorption tests the narrative model; lead--lag is weak.} Estimated absorption predicts strong reversal in all $29$ markets, even where the theory allows continuation. Decomposing it shows why. Only $12$--$19\%$ of its variance is the institution's trade and $2$--$4\%$ is noise trading. Most of the rest is \emph{narrative misfit}: a linear narrative model ignores that a statement's impact scales with the size of its echo cascade. With the true absorption $\lambda(x+u)$, the coefficient is insignificant in both groups, so these experiments do not test \cref{thm:absorption}. They do show that its sign must be estimated. The signed L\'evy area has the predicted sign in $27$ of $29$ markets but is significant in only $9$.

\begin{table}[t]
\centering
\begin{minipage}[t]{0.5\textwidth}
\centering
\caption{Out-of-sample IC, main market ($5$ seeds; oracle $v-p$: $0.181$).}
\label{tab:mainForecast}
\small
\begin{tabular}{@{}lr@{}}
\toprule
Forecast & IC \\
\midrule
Follow the tone & $-0.119\pm0.013$ \\
Follow Say & $-0.047\pm0.006$ \\
Theory composite, unit weights & $0.049\pm0.010$ \\
Fade the price reaction & $0.162\pm0.018$ \\
Boosting: raw voices & $0.109\pm0.015$ \\
Boosting: raw + theory features & $0.118\pm0.019$ \\
Linear: raw voices & $0.169\pm0.017$ \\
Linear: raw + theory features & $\mathbf{0.171\pm0.018}$ \\
\bottomrule
\end{tabular}
\end{minipage}\hfill
\begin{minipage}[t]{0.47\textwidth}
\centering
\caption{Text benchmark ($3$ seeds). Week-ahead analog IC; share of $10$ nearest neighbours with the same consequence (chance $0.25$). Ceiling: $0.936$.}
\label{tab:mainText}
\small
\setlength{\tabcolsep}{4pt}
\begin{tabular}{@{}lrr@{}}
\toprule
Method & IC & Conseq. \\
\midrule
TF-IDF neighbours & $0.192$ & $0.506$ \\
Random 16-d projection & $0.039$ & $0.327$ \\
Ridge on words & $-0.019$ & -- \\
Boosting on words & $0.837$ & -- \\
Return-aligned, 16-d & $\mathbf{0.892}$ & $\mathbf{1.000}$ \\
\bottomrule
\end{tabular}
\end{minipage}
\end{table}

\paragraph{Q4: little forecasting headroom.} Following the tone or institutional statements loses money, and fading the price reaction is a strong baseline (\cref{tab:mainForecast}). Adding theory features to a linear model changes the IC by at most about $0.005$ in any seed. For boosting they help a little, raising the IC from $0.109$ to $0.118$ and improving it in $18$ of $29$ markets. The unit-weight composite is weak ($0.049$), mostly because its absorption term has the wrong sign here. With the oracle at $0.181$, this comparison has little power. It is consistent with \cref{thm:followfade}, which says the optimal forecast is linear in the voices, although the simulated price is only approximately linear in them (Q3). The theory's clearest contributions here are its signs and the diagnostics of Q1 and Q2.

\paragraph{Q5: return-aligned embeddings organise text by consequence.} In a synthetic corpus, each headline mixes irrelevant topic words, tone words and one cue that decides whether the move reverses or continues. We train a 16-dimensional projection of TF-IDF vectors with $\mathcal{L}$ (\cref{app:text}). The ten nearest neighbours of a test headline switch from sharing its topic ($85\%$ under TF-IDF) to sharing its consequence ($100\%$) (\cref{fig:experiments}c). The analog forecast reaches an IC of $0.892$ against a ceiling of $0.936$ (\cref{tab:mainText}). Ridge regression fails because the outcome is a tone$\times$cue interaction. The certificate never made a wrong claim, but it certified nothing under realistic outcome noise: the excess loss per anchor stays at $0.02$--$0.03$ while targets are of order $0.007$. With noise-free outcomes and sharp targets, form (ii) certifies $45\%$ of comparisons, and form (i) still certifies none.

\paragraph{Robustness.} \Cref{tab:mainRobust} changes one assumption at a time. Fat-tailed values, power-law echo delays, partial disclosure and a saturating crowd leave the echo and news signs intact, and detection stays at an AUC of at least $0.857$. The controlled reversed-words test, already insignificant in the reference market, takes the wrong sign in all three seeds under a saturating crowd (\cref{app:results}). Ignoring meaning, or weak meaning ($\varkappa=5$), underestimates the echo share by about six points and weakens the echo test.

\begin{table}[t]
\caption{Robustness (three seeds each; ``reference'' is the main market on the same seeds). IC: linear model with theory features. $t$: firm-clustered. Echo share: estimated (true).}
\label{tab:mainRobust}
\centering
\small
\setlength{\tabcolsep}{5pt}
\begin{tabular}{@{}lrrrrc@{}}
\toprule
Variant & IC & echo $t$ & news $t$ & AUC & echo share \\
\midrule
Reference & $0.190$ & $-7.5$ & $+3.6$ & $0.941$ & $0.425$ ($0.430$) \\
Fat-tailed values ($t_3$) & $0.141$ & $-8.1$ & $+2.4$ & $0.924$ & $0.423$ ($0.428$) \\
Saturating crowd & $0.101$ & $-5.9$ & $+3.2$ & $0.942$ & $0.425$ ($0.430$) \\
Power-law echo delays & $0.177$ & $-7.9$ & $+2.1$ & $0.939$ & $0.417$ ($0.429$) \\
30\% of trade disclosed, noisier & $0.191$ & $-7.5$ & $+4.1$ & $0.857$ & $0.425$ ($0.430$) \\
Weak meaning ($\varkappa=5$) & $0.188$ & $-7.2$ & $+1.9$ & $0.941$ & $0.373$ ($0.429$) \\
Timing only (reference markets) & $0.189$ & $-6.0$ & $+2.6$ & $0.940$ & $0.368$ ($0.430$) \\
\bottomrule
\end{tabular}
\end{table}

\section{Limitations and conclusion}\label{sec:discussion}

\paragraph{Limitations.} The price-formation results are exact in linear--Gaussian economies with reduced-form crowd pricing, and they pin down signs and sufficient statistics, not magnitudes. The simulator shares the theory's structure and is far more predictable than any real market, so our experiments show that the tools \emph{can} work, not that they will. The false-alarm detector needs labels that real data provide only after the fact. Absorption needs a nonlinear narrative model before its sign can be trusted, and the certificate needs outcome targets smoother than raw returns. The text benchmark is synthetic, and the week-ahead outcome is fully determined by two words.

\paragraph{Conclusion.} Treating financial text as a strategic signal changes what should be learned from it. Words whose covariance with deeds is negative should be faded, and so should repetition. The echo prediction held in all $29$ simulated markets, and the Say--Do switch in four of five main-market seeds. The natural next step is a point-in-time study on real disclosures, echo cascades and positioning, testing the echo and Say--Do predictions first.

\paragraph{Broader impact.} These tools are meant for detection, research and investor protection, and they use public information only. They could help regulators and investors discount manipulative narratives. The same understanding could in principle help someone design narratives that move prices, which is market manipulation and illegal. A strategic reading is never evidence that anyone lied (\cref{app:impact}).

\newpage
\appendix
\section*{Appendix}

The appendix contains the full model and the proofs (\cref{app:model,app:strategic,app:measure}), the decision layer (\cref{app:inference}), numerical checks of every closed form (\cref{app:verify}), the complete experimental protocol and results (\cref{app:experiments,app:results}) and a discussion of broader impacts (\cref{app:impact}). Results stated in the main text are restated here with their original numbers.

\section{The three-voice market: model and proofs}\label{app:model}

\subsection{Auxiliary lemmas}
\begin{lemma}[Gaussian conditioning]\label{lem:gauss}
If $(\bm a,\bm b)$ is jointly Gaussian with $\Var(\bm b)\succ0$, then $\E[\bm a\mid\bm b]=\E\bm a+\Cov(\bm a,\bm b)\Var(\bm b)^{-1}(\bm b-\E\bm b)$ and $\Var(\bm a\mid\bm b)=\Var(\bm a)-\Cov(\bm a,\bm b)\Var(\bm b)^{-1}\Cov(\bm b,\bm a)$. The same formulas hold conditionally on a further jointly Gaussian vector.
\end{lemma}
\begin{proof}
Let $\bm e=\bm a-\E\bm a-\Cov(\bm a,\bm b)\Var(\bm b)^{-1}(\bm b-\E\bm b)$. Then $\Cov(\bm e,\bm b)=0$, and joint Gaussianity makes $\bm e$ independent of $\bm b$. Hence $\E[\bm a\mid\bm b]$ is as stated, and $\Var(\bm a\mid\bm b)=\Var(\bm e)$, which expands to the stated formula.
\end{proof}
\begin{lemma}[Bhattacharyya coefficient of Gaussians]\label{lem:bc}
For $p_i=\N(\bm m_i,\Sigma)$ on $\R^d$, $\int\sqrt{p_0p_1}=\exp\big(-\tfrac18(\bm m_1-\bm m_0)^\top\Sigma^{-1}(\bm m_1-\bm m_0)\big)$.
\end{lemma}
\begin{proof}
With $\bar{\bm m}=(\bm m_0+\bm m_1)/2$ and $\bm\delta=\bm m_1-\bm m_0$, the identity $\tfrac12(\bm x-\bm m_0)^\top\Sigma^{-1}(\bm x-\bm m_0)+\tfrac12(\bm x-\bm m_1)^\top\Sigma^{-1}(\bm x-\bm m_1)=(\bm x-\bar{\bm m})^\top\Sigma^{-1}(\bm x-\bar{\bm m})+\tfrac14\bm\delta^\top\Sigma^{-1}\bm\delta$ gives $\sqrt{p_0p_1}=\N(\bm x;\bar{\bm m},\Sigma)\,e^{-\bm\delta^\top\Sigma^{-1}\bm\delta/8}$. Integrate over $\bm x$.
\end{proof}

\subsection{Model}
\begin{assumption}[Fundamentals and public signals]\label{ass:signals}
The asset's terminal value is $v\sim\N(0,\sigma_v^2)$. A vector of $K$ public channels is observed,
\[
  \y = \h\, v + \bm\nu \in\R^K,\qquad \bm\nu\sim\N(\bm 0,\Sigma_\nu)\ \text{independent of } v,
\]
with $\Sigma_y\coloneqq\sigma_v^2\h\h^\top+\Sigma_\nu\succ0$. Channels are labelled by voice: institutional statements (\SAY), independent news, and media repetition (\ECHO).
\end{assumption}
\begin{assumption}[Price formation]\label{ass:price}
Noise traders submit $u\sim\N(0,\sigma_u^2)$, independent of $(v,\bm\nu)$. Informed traders who know $v$ submit $x=\beta(v-p)$ with $\beta>0$. The price is set by a narrative-sensitive crowd and linear impact:
\[
  p=\bphi^\top\y+\lambda\,(x+u),\qquad \lambda>0,
\]
where $\bphi\in\R^K$ collects the crowd's \emph{narrative multipliers}.
\end{assumption}
\begin{figure}[ht]
\centering
\begin{tikzpicture}[
  node distance=1.0cm and 1.5cm,
  latent/.style={circle, draw, thick, minimum size=1.0cm, font=\small},
  obs/.style={rectangle, rounded corners=3pt, draw, thick, minimum width=1.9cm, minimum height=0.8cm, align=center, font=\small},
  arr/.style={-{Stealth[length=2.2mm]}, thick}]
\node[latent, fill=gray!10] (v) {$v$};
\node[obs, draw=say, fill=say!8, text=say, above right=0.6cm and 1.6cm of v] (m) {\textbf{Say} $m$};
\node[obs, draw=do, fill=do!6, text=do, right=1.6cm of v] (n) {News $n$};
\node[obs, draw=echo, fill=echo!8, text=echo, right=1.3cm of m] (e) {\textbf{Echo} $e$};
\node[obs, draw=do, fill=do!8, text=do, below right=0.6cm and 1.6cm of v] (x) {\textbf{Do} $x$};
\node[obs, draw=price, fill=gray!8, minimum width=1.4cm, right=5.3cm of v] (p) {price $p$};
\node[latent, fill=white, right=0.9cm of p] (u) {$u$};
\draw[arr] (v) -- (m);
\draw[arr] (v) -- (n);
\draw[arr] (m) -- (e);
\draw[arr] (v) -- (x);
\draw[arr] (e) -- (p);
\draw[arr] (n) -- (p);
\draw[arr] (m.east) to[out=-10,in=150] (p.north west);
\draw[arr] (x) -- (p);
\draw[arr] (u) -- (p);
\draw[arr, dashed, gray] (p.south) to[out=-110,in=-20] node[below, font=\scriptsize, text=gray] {$x=\beta(v-p)$} (x.south east);
\node[font=\scriptsize, text=echo, above=1pt of e] {garbling of $m$};
\end{tikzpicture}
\caption{Information structure. The echo channel is a downstream garbling of the institutional statement, while news is an independent signal of $v$. Informed demand reacts to the price (dashed), and the price aggregates every channel.}
\label{fig:dag}
\end{figure}
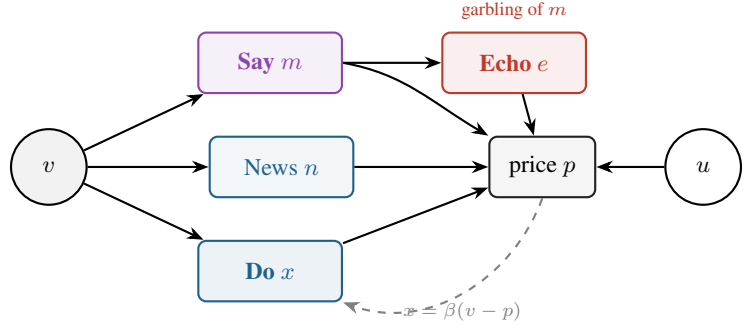
Throughout, $\kappa\coloneqq(1+\lambda\beta)^{-1}\in(0,1)$, the Bayes weights are $\w^\star\coloneqq\sigma_v^2\Sigma_y^{-1}\h$ (so that $\E[v\mid\y]=\w^{\star\top}\y$), and $\sigma^2_{v|y}\coloneqq\Var(v\mid\y)=\sigma_v^2-\sigma_v^4\h^\top\Sigma_y^{-1}\h$. The forward return is $r\coloneqq v-p$.
\begin{lemma}[Price and return]\label{lem:price}
Under \cref{ass:signals,ass:price},
\[
  p=\kappa\big(\bphi^\top\y+\lambda\beta v+\lambda u\big),\qquad r=\kappa\big(v-\bphi^\top\y-\lambda u\big).
\]
\end{lemma}
\begin{proof}
Substituting $x=\beta(v-p)$ gives $p(1+\lambda\beta)=\bphi^\top\y+\lambda\beta v+\lambda u$. Then $r=v-p=\kappa\big((1+\lambda\beta)v-\bphi^\top\y-\lambda\beta v-\lambda u\big)=\kappa(v-\bphi^\top\y-\lambda u)$.
\end{proof}

\subsection{Follow or fade}
\thmfollowfade*
\begin{proof}
By \cref{lem:price}, $r=\kappa(v-\bphi^\top\y-\lambda u)$. Joint Gaussianity gives $\E[v\mid\y]=\Cov(v,\y)\Sigma_y^{-1}\y=\sigma_v^2\h^\top\Sigma_y^{-1}\y=\w^{\star\top}\y$, and $\E[u\mid\y]=0$ by independence. For the variance, $r-\E[r\mid\y]=\kappa\big((v-\E[v\mid\y])-\lambda u\big)$, and the two terms are independent given $\y$.
\end{proof}

\corecho*
\begin{proof}
\emph{Part (i).}
$\sigma(\y)=\sigma(\y_{-e},\xi)$ and $\xi$ is independent of $(v,\y_{-e})$, so $\E[v\mid\y]=\E[v\mid\y_{-e}]$, a linear function of $\y_{-e}$ alone. If $\E[v\mid\y]=\w^{\star\top}\y=\w'^{\top}\y$ with $w'_e=0$, then $(\w^\star-\w')^\top\y=0$ almost surely, so $(\w^\star-\w')^\top\Sigma_y(\w^\star-\w')=0$, and $\Sigma_y\succ0$ forces $\w^\star=\w'$. Hence $w^\star_e=0$, and \cref{thm:followfade} gives the loading $\kappa(0-\varphi_e)$.

\emph{Part (ii).}
Conditionally on $\y_{-n}$, $v\sim\N(\mu',s^2)$ with $\mu'$ linear in $\y_{-n}$, and $n=v+\nu_n$. The scalar Gaussian update gives $\E[v\mid\y]=\mu'+\frac{s^2}{s^2+\sigma_n^2}(n-\mu')$, whose coefficient on $n$ is $w^\star_n$. Apply \cref{thm:followfade}.
\end{proof}

\subsection{The Say--Do gap}
\begin{theorem}[Say--Do gap optimality]\label{thm:saydo}
In addition to $\y$, the econometrician observes a positioning proxy $\hat x=x+\xi_x$, $\xi_x\sim\N(0,\sigma_x^2)$ independent of $(v,\bm\nu,u)$. Let
\[
  c_x\coloneqq\frac{\beta\varsigma^2}{\beta^2\varsigma^2+\sigma_x^2},\qquad \theta\coloneqq\frac{\sigma_x^2}{\beta^2\varsigma^2+\sigma_x^2}\in(0,1].
\]
Then
\begin{enumerate}[label=(\roman*)]
  \item $\E[r\mid\y,\hat x]=\theta\,\kappa(\w^\star-\bphi)^\top\y+c_x\,\hat x$ and $\Var(r\mid\y,\hat x)=\theta\,\varsigma^2<\varsigma^2$ whenever $\sigma_x^2<\infty$;
  \item if the \SAY\ channel $m=y_1$ is overweighted, $\varphi_1>w^\star_1$, the forecast depends on $(m,\hat x)$ only through the \emph{Say--Do gap}
  \[
    G_\gamma\coloneqq\gamma\, m-\hat x,\qquad \gamma\coloneqq\frac{\theta\kappa(\varphi_1-w^\star_1)}{c_x}>0,
  \]
  entering with coefficient $-c_x<0$.
\end{enumerate}
\end{theorem}
\begin{proof}
By \cref{lem:price}, $x=\beta(v-p)=\beta r$, so $\hat x=\beta r+\xi_x$. Conditionally on $\y$, $(r,\hat x)$ is Gaussian with $\Cov(r,\hat x\mid\y)=\beta\varsigma^2$, $\Var(\hat x\mid\y)=\beta^2\varsigma^2+\sigma_x^2$ and $\E[\hat x\mid\y]=\beta\E[r\mid\y]$. Gaussian conditioning (\cref{lem:gauss}) gives
\[
  \E[r\mid\y,\hat x]=\E[r\mid\y]+c_x\big(\hat x-\beta\E[r\mid\y]\big)=(1-\beta c_x)\E[r\mid\y]+c_x\hat x,
\]
with $1-\beta c_x=\theta$, and $\Var(r\mid\y,\hat x)=\varsigma^2-c_x\beta\varsigma^2=\theta\varsigma^2$. For (ii), the terms in $m$ and $\hat x$ are $-\theta\kappa(\varphi_1-w^\star_1)m+c_x\hat x=-c_x(\gamma m-\hat x)$.
\end{proof}
\begin{figure}[ht]
\centering
\begin{tikzpicture}[scale=1.1]
\fill[good!15] (-3,0) rectangle (0,3);
\fill[echo!12] (0,-3) rectangle (3,0);
\fill[gray!7] (0,0) rectangle (3,3);
\fill[gray!7] (-3,-3) rectangle (0,0);
\draw[-{Stealth}, thick] (-3.2,0) -- (3.4,0) node[right, font=\small, align=left] {\SAY\\[-2pt]\scriptsize bullish words};
\draw[-{Stealth}, thick] (0,-3.2) -- (0,3.4) node[above, font=\small, align=center] {\DO\ \scriptsize (buying)};
\node[font=\scriptsize, anchor=east] at (-3.2,0) {bearish words};
\node[font=\scriptsize, anchor=north] at (0,-3.2) {selling};
\draw[dashed, gray!70!black] (-3,-3) -- (3,3) node[above right, font=\scriptsize, text=gray!70!black] {Say = Do};
\node[font=\scriptsize\bfseries, text=good!60!black, align=center] at (-1.6,2.55) {talks bearish, buys};
\node[font=\scriptsize, text=good!60!black, align=center] at (-1.6,2.15) {$\Rightarrow$ bullish signal};
\node[font=\scriptsize\bfseries, text=echo!80!black, align=center] at (1.6,-2.15) {talks bullish, sells};
\node[font=\scriptsize, text=echo!80!black, align=center] at (1.6,-2.55) {$\Rightarrow$ bearish signal};
\node[font=\scriptsize, text=gray!60!black, align=center] at (2.3,0.9) {consistent: follow};
\node[font=\scriptsize, text=gray!60!black, align=center] at (-2.3,-0.9) {consistent: follow};
\foreach \p in {(0.8,1.0),(1.5,1.2),(2.0,2.3),(-1.0,-0.8),(-2.1,-1.7),(0.3,0.6),(-1.3,-1.6),(2.3,1.6),(0.9,-0.6)}
  \fill[gray!60] \p circle (2pt);
\fill[say] (-2.1,1.2) circle (3pt);
\node[font=\scriptsize, text=say, anchor=south] at (-2.1,1.3) {a false alarm};
\draw[{Stealth}-{Stealth}, say, thick] (-2.1,1.2) -- (-0.45,-0.45);
\node[font=\scriptsize, text=say, anchor=east] at (-1.45,0.45) {gap $G$};
\end{tikzpicture}
\caption{The Say--Do plane. \Cref{thm:saydo} shows that the optimal forecast is a decreasing function of the signed distance from the line $\hat x=\gamma m$. Top-left is the false-alarm zone; bottom-right is the false-euphoria zone.}
\label{fig:saydo}
\end{figure}
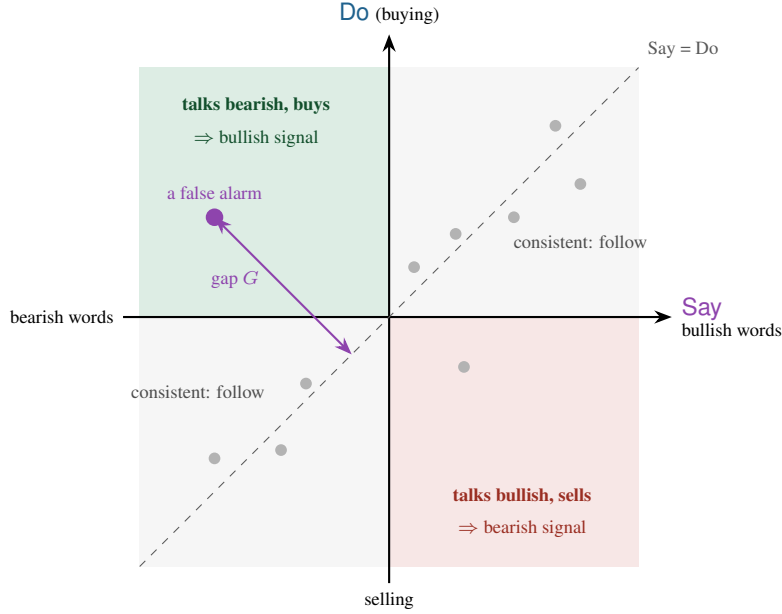

\subsection{Absorption}
The \emph{absorption} of a narrative shock is the price move net of the narrative-implied move, $a\coloneqq p-\bphi^\top\y$.
\begin{theorem}[Absorption sign]\label{thm:absorption}
Under \cref{ass:signals,ass:price},
\[
  \E[r\mid\y,a]=\E[r\mid\y]+c_A\big(a-\E[a\mid\y]\big),\qquad
  c_A=\frac{\beta\,\sigma^2_{v|y}-\lambda\,\sigma_u^2}{\lambda\,\big(\beta^2\sigma^2_{v|y}+\sigma_u^2\big)} .
\]
Absorption therefore predicts \emph{continuation} in its own direction if and only if $\beta\sigma^2_{v|y}>\lambda\sigma_u^2$, and \emph{reversal} if and only if $\beta\sigma^2_{v|y}<\lambda\sigma_u^2$. In the limits, $c_A=1/(\lambda\beta)$ when $\sigma_u=0$ (absorption is sufficient for $r$) and $c_A=-1$ when $\beta=0$ (pure noise fully reverts).
\end{theorem}
\begin{proof}
Write $\tilde v\coloneqq v-\E[v\mid\y]$, independent of $u$ given $\y$. By \cref{lem:price}, $r-\E[r\mid\y]=\kappa(\tilde v-\lambda u)$. Since $a=\lambda(x+u)=\lambda\beta r+\lambda u$,
\[
  a-\E[a\mid\y]=\lambda\beta\kappa(\tilde v-\lambda u)+\lambda u=\lambda\beta\kappa\,\tilde v+\lambda(1-\lambda\beta\kappa)u=\lambda\kappa(\beta\tilde v+u),
\]
using $1-\lambda\beta\kappa=\kappa$. Hence $\Cov(r,a\mid\y)=\lambda\kappa^2(\beta\sigma^2_{v|y}-\lambda\sigma_u^2)$ and $\Var(a\mid\y)=\lambda^2\kappa^2(\beta^2\sigma^2_{v|y}+\sigma_u^2)$, and \cref{lem:gauss} gives $c_A$ as their ratio. The limits follow by substitution.
\end{proof}
\begin{lemma}[Robustness to a misspecified multiplier]\label{lem:robust}
For any $\hat\bphi\in\R^K$, let $\hat a\coloneqq p-\hat\bphi^\top\y$. Then $\sigma(\y,\hat a)=\sigma(\y,a)$ and hence $\E[r\mid\y,\hat a]=\E[r\mid\y,a]$.
\end{lemma}
\begin{proof}
$\hat a=a+(\bphi-\hat\bphi)^\top\y$ is an invertible affine function of $a$ given $\y$, and conversely.
\end{proof}
\begin{remark}[Kyle calibration]
Plugging in the single-period equilibrium intensities of \citet{kyle1985}, $\beta=\sigma_u/\sigma_{v|y}$ and $\lambda=\sigma_{v|y}/(2\sigma_u)$, gives $\beta\sigma^2_{v|y}=\sigma_u\sigma_{v|y}>\tfrac12\sigma_u\sigma_{v|y}=\lambda\sigma_u^2$ and $c_A=\tfrac12$: absorption is informative and predicts continuation.
\end{remark}

\begin{remark}[Relation to volume--return dynamics]
\Cref{thm:absorption} is in the spirit of \citet{llorente2002}, where high-volume returns continue when informed trading dominates and reverse when hedging dominates. The difference is that we condition on the \emph{narrative}: the object is the price move net of what the public story implied. By \cref{lem:robust}, its information content survives misestimation of the crowd's multiplier.
\end{remark}

\section{Strategic speech: model and proofs}\label{app:strategic}

\begin{assumption}[Strategic institution]\label{ass:strategic}
An institution knows $v$, which has an arbitrary distribution with mean zero and variance $\sigma_v^2<\infty$. It publishes a message $m$ and trades $x$. The public observes $\tilde m=m+\varepsilon$, where $\varepsilon$ has mean zero and variance $\sigma_\varepsilon^2$ and is independent of $v$. The price is $p=\varphi\tilde m+\lambda x$, where $\varphi\ge0$ is the crowd's effective credulity and $\lambda>0$ is price impact. Misreporting costs $\tfrac k2(m-v)^2$ with $k>0$ (legal, reputational). The institution maximises
\[
  J(x,m)=\E\big[x(v-p)\mid v\big]-\tfrac k2(m-v)^2=x(v-\varphi m-\lambda x)-\tfrac k2(m-v)^2 .
\]
We write $a\coloneqq2\lambda k$ for the \emph{deterrence--impact product}.
\end{assumption}
\begin{assumption}[Echo amplification]\label{ass:amplify}
Each article moves the price by $\varphi_0$ per unit of message, and the message propagates through a Hawkes cascade with branching ratio $n\in[0,1)$. By \cref{prop:branching}, the expected cascade size is $(1-n)^{-1}$, so the effective credulity is $\varphi=\varphi_0/(1-n)$.
\end{assumption}

\thmspeech*
\begin{proof}
$J$ is quadratic in $(x,m)$ with Hessian $H=\begin{psmallmatrix}-2\lambda&-\varphi\\-\varphi&-k\end{psmallmatrix}$. Since $-2\lambda<0$, $H$ is negative definite iff $\det H=a-\varphi^2>0$. In that case $J$ is strictly concave and its unique maximiser solves the first-order conditions $v-\varphi m-2\lambda x=0$ and $-\varphi x-k(m-v)=0$. Substituting $x=(v-\varphi m)/(2\lambda)$ into the second condition gives $m(a-\varphi^2)=v(a-\varphi)$. Then $x=(1-\varphi\psi)v/(2\lambda)$, and $1-\varphi\psi=a(1-\varphi)/(a-\varphi^2)$ gives $\chi$. If $\det H<0$, $H$ has a positive eigenvalue with eigenvector $\bm d$, and $J(t\bm d)\to\infty$ as $t\to\infty$.

For the taxonomy, take $v>0$.
$\psi>0\iff a>\varphi$ (the denominator is positive); $\psi<1\iff\varphi^2<\varphi\iff\varphi<1$; and $\sgn\chi=\sgn(1-\varphi)$. When $\varphi>1$, $a>\varphi^2>\varphi$, so $\psi>1$.
\end{proof}

\thmidentity*
\begin{proof}
Given $m$, the trade maximises $x(v-\varphi m-\lambda x)$, so $x=(v-\varphi m)/(2\lambda)$, that is, $v-\varphi m=2\lambda x$. Hence $r=v-\varphi(m+\varepsilon)-\lambda x=2\lambda x-\lambda x-\varphi\varepsilon$. Since $x$ and $m$ are functions of $v$, which is independent of $\varepsilon$, we have $\Cov(\varepsilon,\tilde m)=\sigma_\varepsilon^2$, and $\E[\varepsilon\mid x]=0$. The equilibrium covariance is $\Cov(\chi v,\psi v)$, and $\psi\chi<0$ exactly in regimes (iii) and (iv) of \cref{cor:taxonomy}.
\end{proof}

\subsection{Virality}
\begin{proposition}[Manipulation dichotomy]\label{prop:virality}
Under \cref{ass:strategic,ass:amplify} with $\varphi_0\in(0,1)$, call a regime \emph{manipulative} if $\Cov(\text{Say},\text{Do})<0$.
\begin{enumerate}[label=(\roman*)]
  \item If $a<1$, the only reachable manipulative regime is the false alarm. It occurs if and only if
  $\max\{0,\,1-\varphi_0/a\}<n<1-\varphi_0/\sqrt a$.
  \item If $a>1$, the only reachable manipulative regime is exaggeration. It occurs if and only if $1-\varphi_0<n<1-\varphi_0/\sqrt a$.
  \item If $n>1-\varphi_0/\sqrt a$, no interior optimum exists.
\end{enumerate}
\end{proposition}
\begin{proof}
$\varphi=\varphi_0/(1-n)$ increases continuously from $\varphi_0$ to $\infty$ on $[0,1)$. The false alarm requires $\varphi^2<a<\varphi<1$, hence $a<1$. Exaggeration requires $1<\varphi<\sqrt a$, hence $a>1$. The two are therefore mutually exclusive. The conditions $\varphi>a$, $\varphi>1$ and $\varphi<\sqrt a$ are equivalent to $n>1-\varphi_0/a$, $n>1-\varphi_0$ and $n<1-\varphi_0/\sqrt a$ respectively. For $a<1$ we have $\sqrt a<1$, so $\varphi<\sqrt a$ already implies $\varphi<1$. Part (iii) is \cref{thm:speech}.
\end{proof}
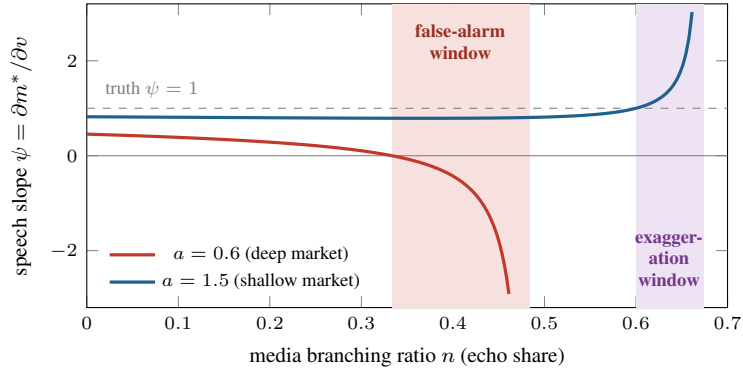
\begin{figure}[ht]
\centering
\begin{tikzpicture}
\begin{axis}[width=0.72\textwidth, height=5.6cm, xmin=0, xmax=0.7, ymin=-3.2, ymax=3.2,
  xlabel={\footnotesize media branching ratio $n$ (echo share)},
  ylabel={\footnotesize speech slope $\psi=\partial m^\ast/\partial v$},
  tick label style={font=\scriptsize}, unbounded coords=jump,
  legend style={font=\scriptsize, draw=none, fill=none, at={(0.02,0.02)}, anchor=south west}]
\fill[echo!15] (axis cs:0.3333,-3.2) rectangle (axis cs:0.4836,3.2);
\fill[say!15] (axis cs:0.6,-3.2) rectangle (axis cs:0.6734,3.2);
\addplot[gray, thin, forget plot] coordinates {(0,0) (0.7,0)};
\addplot[gray, thin, dashed, forget plot] coordinates {(0,1) (0.7,1)};
\addplot[echo, very thick, domain=0:0.478, samples=200, restrict y to domain=-3.3:3.3]
  {(0.6-0.4/(1-x))/(0.6-(0.4/(1-x))^2)};
\addlegendentry{$a=0.6$ (deep market)}
\addplot[do, very thick, domain=0:0.668, samples=200, restrict y to domain=-3.3:3.3]
  {(1.5-0.4/(1-x))/(1.5-(0.4/(1-x))^2)};
\addlegendentry{$a=1.5$ (shallow market)}
\node[font=\scriptsize\bfseries, text=echo!80!black, align=center] at (axis cs:0.408,2.4) {false-alarm\\ window};
\node[font=\scriptsize\bfseries, text=say!80!black, align=center] at (axis cs:0.636,-2.2) {exagger-\\ ation\\ window};
\node[font=\scriptsize, text=gray, anchor=south west] at (axis cs:0.01,1.02) {truth $\psi=1$};
\end{axis}
\end{tikzpicture}
\caption{\Cref{prop:virality} with $\varphi_0=0.4$. As the echo share rises, a deep market (small $a$) passes through a false-alarm window where the speech slope turns negative. A shallow market (large $a$) instead passes through an exaggeration window with $\psi>1$. Both windows end where no interior optimum exists.}
\label{fig:virality}
\end{figure}

\subsection{Adaptive trust}
Suppose that after each round the crowd resets its credulity to the least-squares slope of value on message under the previous round's strategy, a cobweb-type learning rule \citep{ezekiel1938,evans2001}. With noiseless messages $m_t=\psi(\varphi_t)v$, this gives
\[
  \varphi_{t+1}=\frac{\Cov(v,m_t)}{\Var(m_t)}=\frac{1}{\psi(\varphi_t)}=F_a(\varphi_t),\qquad F_a(\varphi)\coloneqq\frac{a-\varphi^2}{a-\varphi},
\]
on the admissible set $\mathcal{D}_a=\{\varphi:\varphi^2<a,\ \varphi\neq a,\ \psi(\varphi)\neq0\}$.

\thmcycles*
\begin{proof}
The theorem collects four facts: (i)~$\varphi=1$ is the unique admissible fixed point of $F_a$ when $a>1$, and there is none when $a\le1$; (ii)~$F_a'(1)=-1/(a-1)$, so truth is locally asymptotically stable iff $a>2$, that is, $\lambda k>1$; (iii)~for $1<a<2$, $F_a$ has exactly one $2$-cycle $\varphi_\pm$, with $\varphi_-<1<\varphi_+$, both points admissible, and multiplier $(5a-9)/(a-1)$, so the cycle is stable iff $\tfrac53<a<2$, that is, $\tfrac56<\lambda k<1$; along it the market alternates between shading ($\varphi_-<1$) and exaggeration ($\varphi_+>1$); (iv)~the multiplier equals $1$ at $a=2$ and $-1$ at $a=\tfrac53$ (flip bifurcations). We prove them in turn.

(i) $F_a(\varphi)=\varphi\iff a-\varphi^2=a\varphi-\varphi^2\iff\varphi=1$. Admissibility of $\varphi=1$ requires $1<a$.
(ii) $F_a'(\varphi)=(a-2a\varphi+\varphi^2)/(a-\varphi)^2$, so $F_a'(1)=(1-a)/(a-1)^2=-1/(a-1)$, and $|F_a'(1)|<1\iff a>2$.
(iii) Using $a-F_a(\varphi)=(a^2-a\varphi-a+\varphi^2)/(a-\varphi)$, a direct computation gives
\[
  F_a(F_a(\varphi))-\varphi=\frac{-a(\varphi-1)\big(2\varphi^2-2a\varphi+a^2-a\big)}{(a-\varphi)^2\,\big(a-F_a(\varphi)\big)} .
\]
The quadratic factor has roots $\varphi_\pm$, which are real and distinct from $1$ for $1<a<2$ (at $\varphi=1$ it equals $(a-1)(a-2)\neq0$). They are not fixed points, so they form the unique $2$-cycle, and $\varphi_-<1<\varphi_+$ because $\varphi_\pm=1$ only at $a\in\{1,2\}$. For admissibility, $\varphi_\pm^2=a\big(1\pm\sqrt{a(2-a)}\big)/2<a$ since $a(2-a)<1$ for $a\neq1$; $\varphi_\pm<a$ since $\sqrt{a(2-a)}<a\iff a>1$; and $\psi(\varphi_\pm)>0$.
For the multiplier, on the cycle $\varphi^2=a\varphi-(a^2-a)/2$, so the numerator of $F_a'$ is $a(3-a-2\varphi)/2$. Vieta's formulas give $\varphi_++\varphi_-=a$ and $\varphi_+\varphi_-=(a^2-a)/2$. Then $(3-a-2\varphi_+)(3-a-2\varphi_-)=5a^2-14a+9=(5a-9)(a-1)$ and $(a-\varphi_+)(a-\varphi_-)=a(a-1)/2$. Hence $F_a'(\varphi_+)F_a'(\varphi_-)=\frac{(a^2/4)(5a-9)(a-1)}{a^2(a-1)^2/4}=\frac{5a-9}{a-1}$. For $1<a<2$, $|5a-9|<a-1\iff\tfrac53<a<2$.
(iv) At $a=2$, $\sqrt{a(2-a)}=0$ and the multiplier equals $1$. At $a=\tfrac53$ it equals $-1$.
\end{proof}

\begin{remark}[Relation to reputation models]
\citet{benabou1992} obtain credibility dynamics from Bayesian updating about a random honest type. Here the dynamics arise from a deterministic cobweb in credulity, and the thresholds $\lambda k=1$ and $\lambda k=\tfrac56$ are explicit. The mechanism is the same as the period-doubling route to chaos in simple nonlinear maps \citep{may1976}. The prediction is new and testable: the rolling Say--Do covariance should be negatively autocorrelated in markets where the product of price impact and misreporting cost is moderate.
\end{remark}

\section{Measurement: proofs and additional results}\label{app:measure}

\subsection{Hawkes processes and declustering}
Let articles about an asset arrive at times $0<\tau_1<\dots<\tau_N\le T$ with sentiment marks $s_1,\dots,s_N$.

\begin{definition}[Linear Hawkes process]
A simple point process on $[0,T]$ with conditional intensity
\[
  \lambda(t)=\mu+\sum_{\tau_l<t}g(t-\tau_l),\qquad \mu>0,\ g\ge0,\ n\coloneqq\int_0^\infty g(s)\,ds<1 .
\]
For the exponential kernel $g(s)=\alpha e^{-\beta_g s}$, the branching ratio is $n=\alpha/\beta_g$.
\end{definition}

By the cluster representation \citep{hawkesoakes1974}, the process is the superposition of an immigrant Poisson stream of rate $\mu$ (\emph{news}) and, for each point $\tau_l$, an offspring Poisson stream of rate $g(\cdot-\tau_l)$ on $(\tau_l,T]$ (\emph{echo}). Let $\pi(j)\in\{0,1,\dots,j-1\}$ denote the parent of article $j$, with $\pi(j)=0$ meaning immigrant.
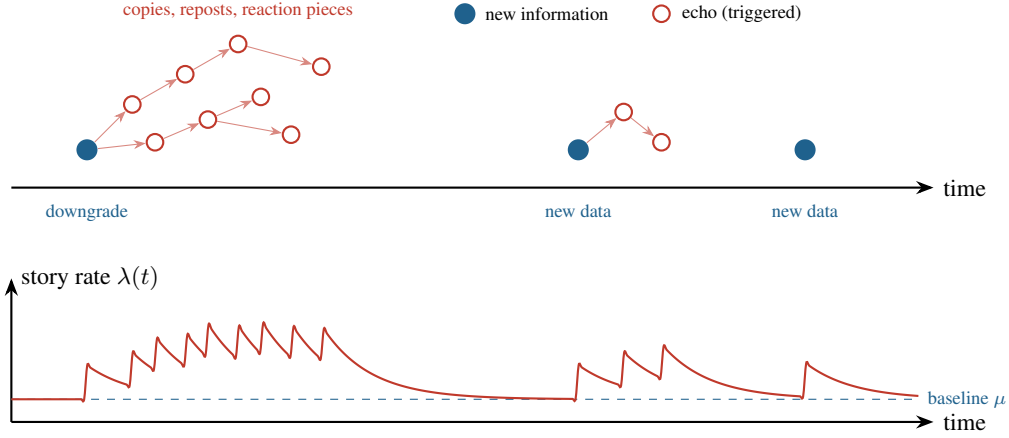
\begin{figure}[ht]
\centering
\begin{tikzpicture}[
  newn/.style={circle, fill=do, inner sep=2.8pt},
  echn/.style={circle, draw=echo, thick, fill=white, inner sep=2.2pt},
  trig/.style={-{Stealth[length=1.6mm]}, echo!60, thin}]
\draw[-{Stealth}, thick] (0,0) -- (12.2,0) node[right, font=\small] {time};
\node[newn] (a) at (1,0.5) {};
\node[echn] (a1) at (1.6,1.1) {};
\node[echn] (a2) at (1.9,0.6) {};
\node[echn] (a3) at (2.3,1.5) {};
\node[echn] (a4) at (2.6,0.9) {};
\node[echn] (a5) at (3.0,1.9) {};
\node[echn] (a6) at (3.3,1.2) {};
\node[echn] (a7) at (3.7,0.7) {};
\node[echn] (a8) at (4.1,1.6) {};
\draw[trig] (a) -- (a1); \draw[trig] (a) -- (a2); \draw[trig] (a1) -- (a3);
\draw[trig] (a2) -- (a4); \draw[trig] (a3) -- (a5); \draw[trig] (a4) -- (a6);
\draw[trig] (a4) -- (a7); \draw[trig] (a5) -- (a8);
\node[newn] (b) at (7.5,0.5) {};
\node[echn] (b1) at (8.1,1.0) {};
\node[echn] (b2) at (8.6,0.6) {};
\draw[trig] (b) -- (b1); \draw[trig] (b1) -- (b2);
\node[newn] (c) at (10.5,0.5) {};
\node[font=\scriptsize, text=do, anchor=north] at (1,-0.1) {downgrade};
\node[font=\scriptsize, text=echo, anchor=south] at (3.0,2.1) {copies, reposts, reaction pieces};
\node[font=\scriptsize, text=do, anchor=north] at (7.5,-0.1) {new data};
\node[font=\scriptsize, text=do, anchor=north] at (10.5,-0.1) {new data};
\begin{scope}[yshift=-3.1cm]
\draw[-{Stealth}, thick] (0,0) -- (12.2,0) node[right, font=\small] {time};
\draw[-{Stealth}, thick] (0,0) -- (0,1.9) node[right, font=\small] {story rate $\lambda(t)$};
\draw[dashed, do] (0,0.3) -- (12,0.3) node[right, font=\scriptsize, text=do] {baseline $\mu$};
\draw[echo, thick, smooth, samples=300, domain=0:12] plot (\x, {0.3
  + 0.45*(\x>1.0)*exp(-1.6*max(\x-1.0,0)) + 0.45*(\x>1.6)*exp(-1.6*max(\x-1.6,0))
  + 0.45*(\x>1.9)*exp(-1.6*max(\x-1.9,0)) + 0.45*(\x>2.3)*exp(-1.6*max(\x-2.3,0))
  + 0.45*(\x>2.6)*exp(-1.6*max(\x-2.6,0)) + 0.45*(\x>3.0)*exp(-1.6*max(\x-3.0,0))
  + 0.45*(\x>3.3)*exp(-1.6*max(\x-3.3,0)) + 0.45*(\x>3.7)*exp(-1.6*max(\x-3.7,0))
  + 0.45*(\x>4.1)*exp(-1.6*max(\x-4.1,0)) + 0.45*(\x>7.5)*exp(-1.6*max(\x-7.5,0))
  + 0.45*(\x>8.1)*exp(-1.6*max(\x-8.1,0)) + 0.45*(\x>8.6)*exp(-1.6*max(\x-8.6,0))
  + 0.45*(\x>10.5)*exp(-1.6*max(\x-10.5,0))});
\end{scope}
\node[newn] at (6.0,2.3) {}; \node[font=\scriptsize, anchor=west] at (6.15,2.3) {new information};
\node[echn] at (8.6,2.3) {}; \node[font=\scriptsize, anchor=west] at (8.75,2.3) {echo (triggered)};
\end{tikzpicture}
\caption{Echo versus news. Top: each new article (blue) triggers a cascade of echoes (red). Bottom: the intensity jumps with every arrival and decays to the baseline $\mu$. \Cref{prop:parentage} assigns each article its exact posterior probability of being news.}
\label{fig:cascade}
\end{figure}
\begin{proposition}[Posterior parentage factorises; cf.\ \citealp{zhuang2002}]\label{prop:parentage}
Conditionally on the full sample $\mathcal{T}=\{\tau_1,\dots,\tau_N\}$ on $[0,T]$, the parent labels $\pi(1),\dots,\pi(N)$ are independent, with
\[
  \Prob\big(\pi(j)=0\mid\mathcal{T}\big)=\frac{\mu}{\lambda(\tau_j)}\eqqcolon p_j,\qquad
  \Prob\big(\pi(j)=l\mid\mathcal{T}\big)=\frac{g(\tau_j-\tau_l)}{\lambda(\tau_j)},\quad l<j .
\]
\end{proposition}
\begin{proof}
Write $\rho_0(t)=\mu$ and $\rho_l(t)=g(t-\tau_l)\ind\{t>\tau_l\}$. In the cluster representation, a labelled configuration is a realisation of the independent Poisson streams, so its density with respect to the reference measure of labelled configurations is the product of the Poisson likelihoods of the streams,
\[
  \prod_{k}\Big[\exp\Big(-\int_0^T\rho_k\Big)\prod_{j:\pi(j)=k}\rho_k(\tau_j)\Big]
  =\exp\Big(-\int_0^T\lambda(t)\,dt\Big)\prod_{j=1}^N\rho_{\pi(j)}(\tau_j),
\]
because $\sum_k\rho_k(t)=\lambda(t)$. Summing over labellings, the product factorises over $j$, and $\sum_k\rho_k(\tau_j)=\lambda(\tau_j)$ recovers the standard Hawkes likelihood $e^{-\int\lambda}\prod_j\lambda(\tau_j)$. Dividing, the posterior of the labelling is $\prod_j\rho_{\pi(j)}(\tau_j)/\lambda(\tau_j)$, a product of independent categorical laws.
\end{proof}
\begin{remark}
The posterior uses only the past of $\tau_j$, although it conditions on the whole sample: future arrivals are informative about \emph{their} parents, not about $\pi(j)$. The echo split is therefore causal and free of look-ahead \citep[cf.][]{zhuang2002}.
\end{remark}
\begin{proposition}[Branching ratio is the long-run echo share]\label{prop:branching}
For the stationary process ($n<1$), the mean intensity is $\Lambda=\mu/(1-n)$, and the long-run fraction of articles that are echoes equals $n$.
\end{proposition}
\begin{proof}
Each immigrant founds a Galton--Watson cluster whose generations have expected sizes $1,n,n^2,\dots$, so its expected total size is $(1-n)^{-1}$. Immigrants arrive at rate $\mu$, so articles arrive at rate $\Lambda=\mu/(1-n)$ (equivalently, stationarity gives $\Lambda=\mu+\Lambda\int g$). The immigrant fraction is $\mu/\Lambda=1-n$.
\end{proof}
\begin{proposition}[MMSE news--echo split]\label{prop:mmse}
Assume that, given $\mathcal{T}$, the marks are independent of the labels. Then the minimum-mean-square-error estimates of the news and echo sentiment totals are
\[
  s^{\text{new}}\coloneqq\E\Big[\sum_j\ind\{\pi(j)=0\}s_j\,\Big|\,\mathcal{T},\bm s\Big]=\sum_j p_j s_j,\qquad
  s^{\text{echo}}=\sum_j(1-p_j)s_j .
\]
\end{proposition}
\begin{proof}
Conditional expectation minimises mean-square error. By linearity and the assumption, $\E[\ind\{\pi(j)=0\}s_j\mid\mathcal{T},\bm s]=s_j\Prob(\pi(j)=0\mid\mathcal{T})=p_js_j$, using \cref{prop:parentage}.
\end{proof}

\subsection{Semantic declustering}
\begin{assumption}[Semantically marked cascade]\label{ass:marked}
In the cluster representation, immigrant articles carry embeddings $\z\sim f_0$ independently, and an offspring of article $l$ carries $\z\sim f(\cdot\mid\z_l)$, for example the von Mises--Fisher density $f(\z\mid\z_l)=C_D(\varkappa)\,e^{\varkappa\langle\z,\z_l\rangle}$ on $S^{D-1}$. Marks are conditionally independent given the parent labels.
\end{assumption}

\thmsemantic*
\begin{proof}
As in \cref{prop:parentage}, the density of a labelled configuration is the product of the stream likelihoods times the mark densities, $e^{-\int_0^T\lambda}\prod_j\rho_{\pi(j)}(\tau_j)\,f_{\pi(j)}(\z_j)$, where $f_0$ is the immigrant mark density and $f_l=f(\cdot\mid\z_l)$. This factorises over $j$; summing over labellings gives $e^{-\int\lambda}\prod_j\Lambda_j$, and the posterior is the product of the stated categorical laws. For (i), marks are now conditioned on, so no independence assumption is needed: $\E[\ind\{\pi(j)=0\}s_j\mid\mathcal{T},\z]=s_j\,p^{\text{sem}}_j$ whenever $s_j$ is a function of $\z_j$, as it is for $s_j=\langle\w_s^\perp,\z_j\rangle$. For (ii), integrating the marks out sequentially shows that the timing-only posterior is that of \cref{prop:parentage}, and conditioning does not increase expected entropy.
\end{proof}
\begin{figure}[ht]
\centering
\begin{tikzpicture}[font=\small]
\draw[-{Stealth}, thick] (0,0) -- (9.0,0) node[right, font=\scriptsize] {time};
\node[circle, fill=do, inner sep=2.6pt, label={[font=\scriptsize, align=center]above:{B: earnings warning\\ \textit{6 hours ago}}}] (B) at (0.9,0) {};
\node[circle, fill=gray!60, inner sep=2.6pt, label={[font=\scriptsize, align=center]above:{A: unrelated\\ macro story\\ \textit{20 min ago}}}] (A) at (4.6,0) {};
\node[circle, draw=echo, very thick, fill=white, inner sep=2.4pt, label={[font=\scriptsize, align=center, text=echo]above:{C: new article,\\ near-copy of B}}] (C) at (7.4,0) {};
\draw[-{Stealth[length=2mm]}, thick, gray, dashed] (A) to[out=-35,in=-145] node[below, font=\scriptsize, text=gray] {timing} (C);
\draw[-{Stealth[length=2mm]}, very thick, echo] (B) to[out=-30,in=-150] node[below, font=\scriptsize, text=echo] {meaning} (C);
\begin{scope}[xshift=10.9cm, yshift=-1.3cm, scale=1.15]
\node[font=\scriptsize\bfseries, anchor=south] at (1.1,2.35) {Parent of C};
\foreach \lab/\x in {news/0, A/0.75, B/1.5} {\node[font=\scriptsize, anchor=north] at (\x+0.3,0) {\lab};}
\draw (0,0) -- (2.4,0);
\fill[gray!45] (0.05,0) rectangle (0.3,0.60);  \fill[echo!80] (0.3,0) rectangle (0.55,0.10);
\fill[gray!45] (0.80,0) rectangle (1.05,1.10); \fill[echo!80] (1.05,0) rectangle (1.30,0.14);
\fill[gray!45] (1.55,0) rectangle (1.80,0.30); \fill[echo!80] (1.80,0) rectangle (2.05,1.76);
\fill[gray!45] (0,-0.75) rectangle (0.2,-0.6); \node[font=\scriptsize, anchor=west] at (0.2,-0.675) {timing only};
\fill[echo!80] (0,-1.05) rectangle (0.2,-0.9); \node[font=\scriptsize, anchor=west] at (0.2,-0.975) {timing + meaning};
\end{scope}
\end{tikzpicture}
\caption{Semantic declustering (\cref{thm:semantic}), illustrative numbers. Timing alone attributes article C to the most recent story A. Adding embedding similarity attributes it to B, of which it is a near-copy. C is thus correctly classified as echo of B, and its sentiment enters $s^{\text{echo}}$.}
\label{fig:semantic}
\end{figure}
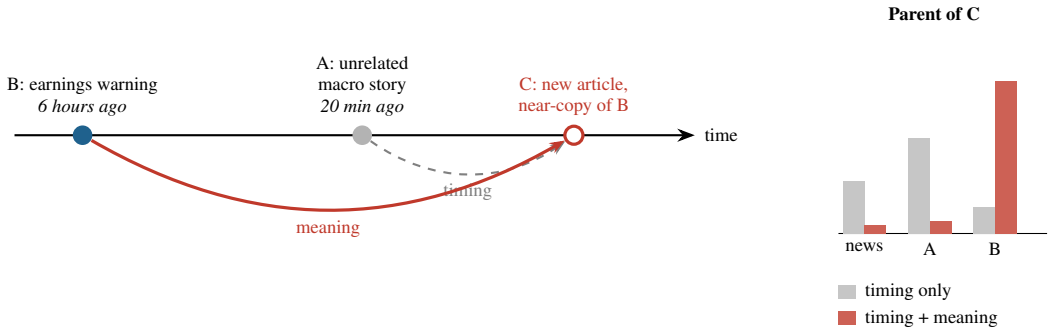

\subsection{Return-aligned embeddings}
\thmracl*
\begin{proof}
For each anchor, $-\sum_kq_{jk}\log p_{jk}=H(q_{j\cdot})+\KL(q_{j\cdot}\|p_{j\cdot})\ge H(q_{j\cdot})$ by Gibbs' inequality, with equality iff $p_{j\cdot}=q_{j\cdot}$. Two softmax vectors coincide iff their logits differ by a constant, so equality holds iff $S_{jk}/\tau+D_{jk}/2b^2=c_j/\tau$ for all $k\neq j$. By symmetry of $S$ and $D$, for every pair $j\neq k$, $c_j=S_{jk}+\tfrac{\tau}{2b^2}D_{jk}=S_{kj}+\tfrac{\tau}{2b^2}D_{kj}=c_k$, so all $c_j$ equal a common $c$. The converse is immediate. On the sphere, $\|\z_j-\z_k\|^2=2-2S_{jk}$.
The neighbour statement follows because $\|\z_j-\z_k\|^2$ is an increasing affine function of $\|\bm\rho_j-\bm\rho_k\|^2$.
\end{proof}

\thmrobustnn*
\begin{proof}
Throughout, $p_{j\cdot}$ is a softmax of $S_{j\cdot}/\tau$, so $p_{jk}>p_{jl}$ is equivalent to $S_{jk}>S_{jl}$, and on the sphere $\|\z_j-\z_k\|^2=2-2S_{jk}$. It therefore suffices to show $p_{jk}>p_{jl}$.

(i) By Pinsker's inequality, $\sum_{k}|p_{jk}-q_{jk}|\le\sqrt{2\epsilon_j}$. Hence $p_{jk}-p_{jl}\ge(q_{jk}-q_{jl})-|p_{jk}-q_{jk}|-|p_{jl}-q_{jl}|>0$.

(ii) Merge all outcomes other than $k$ and $l$ into one cell. By the data-processing inequality, the three-cell divergence satisfies $\KL_3(q\|p)\le\epsilon_j$. Consider the set $\{p: p_k\le p_l\}$. The divergence is convex in $p$, and its unconstrained minimiser $p=q$ lies outside this set because $q_{jk}>q_{jl}$. The constrained minimum is therefore attained on the boundary $p_k=p_l=t$, with $p_{\text{rest}}=1-2t$. Setting the derivative in $t$ to zero gives $t=(q_{jk}+q_{jl})/2$ and $p_{\text{rest}}=q_{\text{rest}}$, where the divergence equals $\Delta_{jkl}$. So every $p$ with $p_{jk}\le p_{jl}$ has $\KL_3(q\|p)\ge\Delta_{jkl}>\epsilon_j$, which is impossible. The boundary point with $p_{\text{rest}}$ split in proportion to $q$ shows that the bound is attained, which gives tightness.

For (i)$\Rightarrow$(ii), write $\pi=q_{jk}/(q_{jk}+q_{jl})$. Then $\Delta_{jkl}=(q_{jk}+q_{jl})\,\KL\big(\text{Bern}(\pi)\,\|\,\text{Bern}(\tfrac12)\big)\ge 2(q_{jk}+q_{jl})(\pi-\tfrac12)^2=\frac{(q_{jk}-q_{jl})^2}{2(q_{jk}+q_{jl})}\ge\frac{(q_{jk}-q_{jl})^2}{2}$, by Pinsker for two points and $q_{jk}+q_{jl}\le1$.
\end{proof}

\begin{remark}[Realisability and regularisation]
Equality requires the matrix with unit diagonal and off-diagonal entries $c-\frac{\tau}{2b^2}D_{jk}$ to be positive semidefinite of rank at most $D$. Otherwise the minimiser is the closest realisable geometry in the KL sense. In practice we add a semantic anchor $\sum_j\KL(p^{(0)}_{j\cdot}\|p_{j\cdot})$ to the frozen encoder, which interpolates between topic geometry and outcome geometry.
\end{remark}
\begin{figure}[ht]
\centering
\begin{tikzpicture}[scale=1.12,
  cr/.style={circle, fill=echo!80, inner sep=2pt},
  cg/.style={circle, fill=good!80, inner sep=2pt},
  sr/.style={rectangle, fill=echo!80, inner sep=2.2pt},
  sg/.style={rectangle, fill=good!80, inner sep=2.2pt}]
\node[font=\small\bfseries] at (2.3,3.4) {Generic embedding};
\node[font=\scriptsize, text=gray!70!black] at (2.3,3.05) {organised by \emph{topic}};
\draw[dashed, gray] (1.3,2.2) ellipse (0.95 and 0.6);
\draw[dashed, gray] (3.5,0.9) ellipse (0.95 and 0.6);
\node[font=\scriptsize, text=gray!70!black] at (1.3,1.4) {oil stories};
\node[font=\scriptsize, text=gray!70!black] at (3.5,0.1) {bank stories};
\foreach \p in {(1.0,2.5),(1.6,2.0),(1.2,1.9)} \node[cr] at \p {};
\foreach \p in {(1.5,2.5),(0.8,2.1),(1.9,2.3)} \node[cg] at \p {};
\foreach \p in {(3.2,1.2),(3.8,0.7),(3.4,0.6)} \node[sr] at \p {};
\foreach \p in {(3.7,1.2),(3.0,0.8),(4.0,1.0)} \node[sg] at \p {};
\draw[-{Stealth[length=3mm]}, very thick, say] (4.9,2.3) -- (6.6,2.3) node[midway, above, font=\scriptsize, align=center] {fine-tune on\\ market outcomes};
\node[font=\small\bfseries] at (9.0,3.4) {Market-aligned embedding};
\node[font=\scriptsize, text=gray!70!black] at (9.0,3.05) {organised by \emph{consequence}};
\draw[dashed, echo] (7.9,2.2) ellipse (0.95 and 0.6);
\draw[dashed, good] (10.1,0.9) ellipse (0.95 and 0.6);
\node[font=\scriptsize, text=echo, anchor=north] at (7.9,1.55) {news that got overdone};
\node[font=\scriptsize, text=good!70!black] at (10.1,0.1) {news that stuck};
\foreach \p in {(7.5,2.4),(8.2,2.0),(7.8,1.9)} \node[cr] at \p {};
\foreach \p in {(8.1,2.5),(7.4,2.0),(8.4,2.3)} \node[sr] at \p {};
\foreach \p in {(9.8,1.2),(10.4,0.7),(10.0,0.6)} \node[cg] at \p {};
\foreach \p in {(10.3,1.2),(9.6,0.8),(10.6,1.0)} \node[sg] at \p {};
\node[cr] at (1.0,-0.6) {}; \node[font=\scriptsize, anchor=west] at (1.15,-0.6) {price later reversed};
\node[cg] at (4.0,-0.6) {}; \node[font=\scriptsize, anchor=west] at (4.15,-0.6) {price kept moving};
\node[circle, draw, inner sep=2pt] at (7.0,-0.6) {}; \node[font=\scriptsize, anchor=west] at (7.15,-0.6) {oil story};
\node[rectangle, draw, inner sep=2.2pt] at (9.0,-0.6) {}; \node[font=\scriptsize, anchor=west] at (9.15,-0.6) {bank story};
\end{tikzpicture}
\caption{\Cref{thm:racl} in pictures. Before fine-tuning, stories cluster by subject. At the optimum, distances mirror market aftermaths, so ``news that got overdone'' becomes a region of the map.}
\label{fig:finetune}
\end{figure}
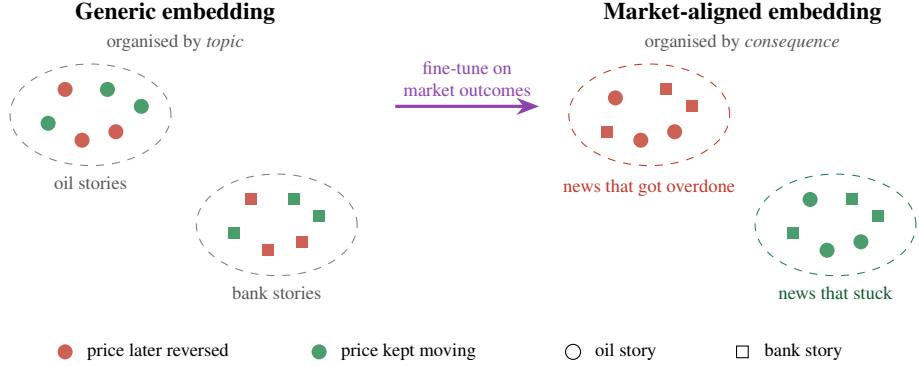

\subsection{Topic-orthogonal sentiment, surprise and analog forecasting}
\begin{lemma}[Topic invariance]\label{lem:topic}
Let $U\in\R^{D\times r}$ have orthonormal columns spanning topic directions and $\w_s^\perp\coloneqq(I-UU^\top)\w_s$. Then (i) $\langle\w_s^\perp,\z+U\bm t\rangle=\langle\w_s^\perp,\z\rangle$ for all $\bm t\in\R^r$, and (ii) $\w_s^\perp=\arg\min\{\|\w-\w_s\|:U^\top\w=\bm 0\}$.
\end{lemma}
\begin{proof}
(i) $U^\top\w_s^\perp=U^\top\w_s-U^\top UU^\top\w_s=\bm 0$. (ii) $\w_s^\perp$ is the orthogonal projection of $\w_s$ onto $\ker U^\top$.
\end{proof}
A causal sequence model predicts the next embedding $\hat\z_t$ from past coverage and market state; the innovation is $\bm\varepsilon_t=\z_t-\hat\z_t$.
\begin{proposition}[Calibrated surprise]\label{prop:surprise}
If the predictive law is $\z_t\mid\F_{t^-}\sim\N(\hat\z_t,\Sigma)$ with $\Sigma\succ0$, then $S_t\coloneqq\bm\varepsilon_t^\top\Sigma^{-1}\bm\varepsilon_t\sim\chi^2_D$ and $s_t/\sqrt{\w_s^{\perp\top}\Sigma\w_s^\perp}\sim\N(0,1)$ with $s_t\coloneqq\langle\w_s^\perp,\bm\varepsilon_t\rangle$. Moreover $S_t=-2\log p(\z_t\mid\F_{t^-})-D\log2\pi-\log\det\Sigma$: Mahalanobis surprise is Shannon surprisal up to constants.
\end{proposition}
\begin{proof}
$\Sigma^{-1/2}\bm\varepsilon_t\sim\N(\bm 0,I_D)$, so $S_t$ is a sum of $D$ squared independent standard normals. The projection of a Gaussian vector is Gaussian with variance $\w^\top\Sigma\w$. The last identity is the Gaussian log-density.
\end{proof}
\begin{definition}
Given a story $\z$ and a market state $\bm\zeta$, the \emph{Market Interpretation Function} $\mathcal{M}(\z,\bm\zeta)$ is the conditional law of the abnormal return over the reaction window, with mean $\mu(\z,\bm\zeta)$.
\end{definition}

The \emph{d\'ej\`a vu} estimator averages the reactions of historical analogues with weights $k_j\ge0$, normalised as $\omega_j=k_j/\sum_lk_l$: $\hat\mu(\z)=\sum_j\omega_j\mathrm{AR}_j$.
\begin{proposition}[Bias--variance of analog forecasting]\label{prop:analog}
Suppose $\mathrm{AR}_j=\mu(\z_j)+e_j$ with $\E[e_j\mid\z_{1:N}]=0$, conditionally uncorrelated errors of variance at most $\sigma^2$, and $\mu$ $L$-Lipschitz for a metric $d$. Then
\[
  \big|\E[\hat\mu(\z)\mid\z_{1:N}]-\mu(\z)\big|\le L\sum_j\omega_jd(\z,\z_j),
\]
\[
  \Var\big(\hat\mu(\z)\mid\z_{1:N}\big)\le\frac{\sigma^2}{N_{\text{eff}}},\quad N_{\text{eff}}\coloneqq\Big(\sum_j\omega_j^2\Big)^{-1}.
\]
\end{proposition}
\begin{proof}
The bias is $\sum_j\omega_j(\mu(\z_j)-\mu(\z))$, bounded termwise by Lipschitz continuity. The variance is $\sum_j\omega_j^2\Var(e_j\mid\cdot)\le\sigma^2\sum_j\omega_j^2$.
\end{proof}

\subsection{L\'evy area}
\begin{definition}[L\'evy area]
For $C^1$ paths $X,Y:[0,T]\to\R$,
\[
  \mathcal{A}_T(X,Y)=\tfrac12\int_0^T(X_u-X_0)\,dY_u-(Y_u-Y_0)\,dX_u=\tfrac12\big(S^{XY}_T-S^{YX}_T\big),
\]
where $S^{XY}_T=\int_{0<s<t<T}dX_s\,dY_t$ is a level-two signature term \citep{lyons1998,chevyrev2016}.
\end{definition}
\begin{proposition}[Structural properties]\label{prop:levyprops}
(i) $\mathcal{A}(X,Y)=-\mathcal{A}(Y,X)$; (ii) $\mathcal{A}(aX+c,bY+d)=ab\,\mathcal{A}(X,Y)$; (iii) $\mathcal{A}$ is invariant under $C^1$ increasing reparametrisations of time; (iv) if $(X,Y)$ traces a positively oriented simple closed curve, $\mathcal{A}$ equals the enclosed area; (v) for a piecewise-linear path through $(X_i,Y_i)_{i=0}^n$,
\[
  \mathcal{A}=\tfrac12\sum_{i=0}^{n-1}\big[(X_i-X_0)(Y_{i+1}-Y_i)-(Y_i-Y_0)(X_{i+1}-X_i)\big].
\]
\end{proposition}
\begin{proof}
(i)--(ii) are immediate from the definition. (iii) follows from the change of variables $u=\psi(s)$ in both integrals. (iv) is Green's theorem applied to $\tfrac12(x\,dy-y\,dx)$. (v): on the segment from $i$ to $i+1$, $X_u-X_0=(X_i-X_0)+t\Delta X$ and $Y_u-Y_0=(Y_i-Y_0)+t\Delta Y$ for $t\in[0,1]$. The segment contributes $\tfrac12[(X_i-X_0)\Delta Y-(Y_i-Y_0)\Delta X]+\tfrac12\int_0^1(t\Delta X\Delta Y-t\Delta Y\Delta X)\,dt$, and the last integral vanishes.
\end{proof}

\thmlevy*
\begin{proof}
Mean-square differentiability gives $\E[X_s\dot Y_t]=\partial_tC(t-s)=C'(t-s)$ and, since $\E[Y_sX_t]=C(s-t)$, $\E[Y_s\dot X_t]=-C'(s-t)$. Therefore
\[
  \E[(X_u-X_0)\dot Y_u]=C'(0)-C'(u),\qquad \E[(Y_u-Y_0)\dot X_u]=-C'(0)+C'(-u).
\]
The integrands are integrable by Cauchy--Schwarz and stationarity, so Fubini's theorem gives
\[
  \E[\mathcal{A}_T]=\tfrac12\int_0^T\big(2C'(0)-C'(u)-C'(-u)\big)\,du=T\,C'(0)-\tfrac12\big(C(T)-C(0)\big)-\tfrac12\big(C(0)-C(-T)\big),
\]
which simplifies to the claim.

For the lag statement:
$C(h)=\E[X_s(\gamma X_{s+h-\ell}+Z_{s+h})]=\gamma R(h-\ell)$, so $C'(0)=\gamma R'(-\ell)=-\gamma R'(\ell)$ by evenness. If $\ell>0$ then $R'(\ell)<0$; if $\ell<0$ then $R'(\ell)=-R'(-\ell)>0$.
Since $\E[\mathcal{A}_T]/T\to C'(0)$ when $C$ is bounded, the long-run area rate has the sign of $-\gamma R'(\ell)$, that is, of $\gamma\ell$.
\end{proof}
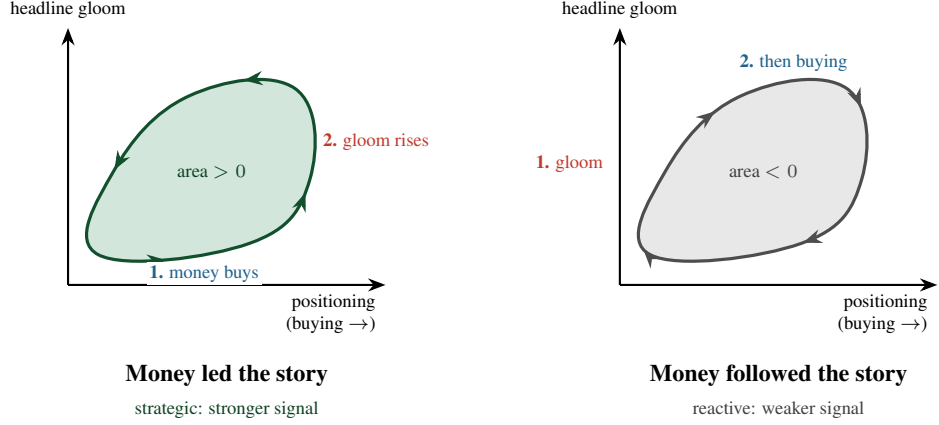
\begin{figure}[ht]
\centering
\begin{tikzpicture}[scale=1.0,
  arrowmid/.style={postaction={decorate}, decoration={markings,
    mark=at position 0.12 with {\arrow{Stealth[length=2.5mm]}},
    mark=at position 0.37 with {\arrow{Stealth[length=2.5mm]}},
    mark=at position 0.62 with {\arrow{Stealth[length=2.5mm]}},
    mark=at position 0.87 with {\arrow{Stealth[length=2.5mm]}}}}]
\begin{scope}
\draw[-{Stealth}, thick] (0,0) -- (4.2,0) node[below left, font=\scriptsize, align=right] {positioning\\(buying $\rightarrow$)};
\draw[-{Stealth}, thick] (0,0) -- (0,3.4) node[above, font=\scriptsize] {headline gloom};
\fill[good!20] plot[smooth cycle, tension=0.7] coordinates {(0.4,0.4) (2.4,0.55) (3.2,1.4) (3.0,2.6) (1.6,2.5) (0.6,1.5)};
\draw[good!60!black, very thick, arrowmid] plot[smooth cycle, tension=0.7] coordinates {(0.4,0.4) (2.4,0.55) (3.2,1.4) (3.0,2.6) (1.6,2.5) (0.6,1.5)};
\node[font=\scriptsize, text=good!50!black] at (1.9,1.5) {area $>0$};
\node[font=\scriptsize, anchor=north, text=do, fill=white, inner sep=1pt] at (1.8,0.3) {\textbf{1.} money buys};
\node[font=\scriptsize, anchor=west, text=echo] at (3.25,1.9) {\textbf{2.} gloom rises};
\node[font=\small\bfseries, align=center] at (2.1,-1.2) {Money led the story};
\node[font=\scriptsize, text=good!50!black] at (2.1,-1.65) {strategic: stronger signal};
\end{scope}
\begin{scope}[xshift=7.3cm]
\draw[-{Stealth}, thick] (0,0) -- (4.2,0) node[below left, font=\scriptsize, align=right] {positioning\\(buying $\rightarrow$)};
\draw[-{Stealth}, thick] (0,0) -- (0,3.4) node[above, font=\scriptsize] {headline gloom};
\fill[gray!18] plot[smooth cycle, tension=0.7] coordinates {(0.4,0.4) (0.6,1.5) (1.6,2.5) (3.0,2.6) (3.2,1.4) (2.4,0.55)};
\draw[gray!60!black, very thick, arrowmid] plot[smooth cycle, tension=0.7] coordinates {(0.4,0.4) (0.6,1.5) (1.6,2.5) (3.0,2.6) (3.2,1.4) (2.4,0.55)};
\node[font=\scriptsize, text=gray!50!black] at (1.9,1.5) {area $<0$};
\node[font=\scriptsize, anchor=east, text=echo] at (-0.1,1.6) {\textbf{1.} gloom};
\node[font=\scriptsize, anchor=south, text=do] at (2.3,2.7) {\textbf{2.} then buying};
\node[font=\small\bfseries, align=center] at (2.1,-1.2) {Money followed the story};
\node[font=\scriptsize, text=gray!50!black] at (2.1,-1.65) {reactive: weaker signal};
\end{scope}
\end{tikzpicture}
\caption{\Cref{thm:levy}: the orientation of the loop in the (positioning, gloom) plane reveals who moved first, and its enclosed area (\cref{prop:levyprops}(iv)) measures how strongly.}
\label{fig:loop}
\end{figure}

\section{From signals to positions}\label{app:inference}\label{sec:inference}\label{sec:sizing}
Events belong to latent regimes $K_t\in\{\mathrm{I},\mathrm{O},\mathrm{S}\}$ (informational, overreaction, strategic) following a Markov chain, with emissions $\bm o_t\mid K_t=k\sim f_k$, where $\bm o_t$ collects $(s^{\text{new}},s^{\text{echo}},G,A,\mathcal{A},\dots)$.

\begin{proposition}[Causal filtering and mixture optimality]\label{prop:filter}
The filtered probabilities $\pi_t(k)=\Prob(K_t=k\mid\F_t)$ are $\F_t$-measurable and satisfy $\pi_t\propto f(\bm o_t)\odot P^\top\pi_{t-1}$. The MSE-optimal forecast is the mixture of experts
\[
  \E[r_{t+h}\mid\F_t]=\sum_k\pi_t(k)\,\E[r_{t+h}\mid\F_t,K_t=k].
\]
Smoothed probabilities $\Prob(K_t=k\mid\F_T)$, $T>t$, are not $\F_t$-measurable in general and must not be used for forecasting.
\end{proposition}
\begin{proof}
The recursion is Bayes' rule with prediction step $P^\top\pi_{t-1}$. The mixture follows from the tower property. Smoothed probabilities depend on $\bm o_{t+1:T}$.
\end{proof}

\begin{theorem}[Exponential detectability of strategic episodes]\label{thm:detect}
Consider an episode of $n$ events that are i.i.d.\ given a fixed regime, with $\bm o\sim\N(\bm m_0,\Sigma)$ under the null regime and $\N(\bm m_1,\Sigma)$ under the strategic regime, and priors $\pi_0,\pi_1$. The Bayes (MAP) misclassification probability satisfies
\[
  P_e\le\sqrt{\pi_0\pi_1}\;e^{-n\Delta^2/8}\le\tfrac12e^{-n\Delta^2/8},\qquad \Delta^2\coloneqq(\bm m_1-\bm m_0)^\top\Sigma^{-1}(\bm m_1-\bm m_0).
\]
Consequently $n\ge(8/\Delta^2)\log\big(1/(2\epsilon)\big)$ events suffice for $P_e\le\epsilon$.
\end{theorem}
\begin{proof}
$P_e=\int\min\{\pi_0p_0^{\otimes n},\pi_1p_1^{\otimes n}\}\le\sqrt{\pi_0\pi_1}\int\sqrt{p_0^{\otimes n}p_1^{\otimes n}}=\sqrt{\pi_0\pi_1}\,\mathrm{BC}^n$, using $\min\{a,b\}\le\sqrt{ab}$, where $\mathrm{BC}=\int\sqrt{p_0p_1}$ is the Bhattacharyya coefficient \citep{bhattacharyya1943,kailath1967}. For equal-covariance Gaussians, completing the square gives $\mathrm{BC}=e^{-\Delta^2/8}$ (\cref{lem:bc}). Finally $\sqrt{\pi_0\pi_1}\le\tfrac12$.
\end{proof}

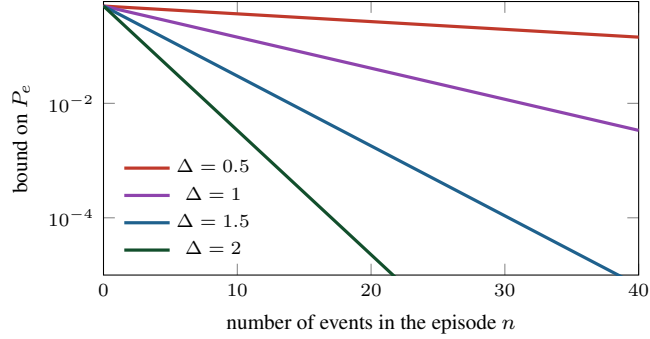
\begin{figure}[ht]
\centering
\begin{tikzpicture}
\begin{semilogyaxis}[width=0.62\textwidth, height=5.2cm, xmin=0, xmax=40, ymin=1e-5, ymax=0.6,
  xlabel={\footnotesize number of events in the episode $n$},
  ylabel={\footnotesize bound on $P_e$},
  legend style={font=\scriptsize, draw=none, fill=none, at={(0.02,0.02)}, anchor=south west},
  tick label style={font=\scriptsize}, samples=80, domain=0:40]
\addplot[echo, very thick] {0.5*exp(-x*0.25/8)}; \addlegendentry{$\Delta=0.5$}
\addplot[say, very thick] {0.5*exp(-x*1/8)}; \addlegendentry{$\Delta=1$}
\addplot[do, very thick] {0.5*exp(-x*2.25/8)}; \addlegendentry{$\Delta=1.5$}
\addplot[good!60!black, very thick] {0.5*exp(-x*4/8)}; \addlegendentry{$\Delta=2$}
\end{semilogyaxis}
\end{tikzpicture}
\caption{\Cref{thm:detect}. Detection error falls exponentially in the length of a strategic episode, at a rate set by the Mahalanobis separation $\Delta$ of the regime fingerprints.}
\label{fig:detect}
\end{figure}

Forecast intervals use adaptive conformal inference \citep{gibbs2021}. Let $\hat C_t(\alpha')$ be nested prediction sets with $\hat C_t(\alpha')=\R$ for $\alpha'\le0$ and $\hat C_t(\alpha')=\emptyset$ for $\alpha'\ge1$. Set $\mathrm{err}_t=\ind\{r_t\notin\hat C_t(\alpha_t)\}$ and $\alpha_{t+1}=\alpha_t+\gamma(\alpha-\mathrm{err}_t)$ with $\alpha_1\in[0,1]$ and $\gamma>0$.

\begin{theorem}[Deterministic long-run coverage; \citealp{gibbs2021}, Prop.~4.1]\label{thm:aci}
For \emph{every} sequence of outcomes, including adversarial or regime-switching ones,
\[
  \Big|\frac1T\sum_{t=1}^T\mathrm{err}_t-\alpha\Big|\le\frac{\max\{\alpha_1,1-\alpha_1\}+\gamma}{\gamma\,T}.
\]
\end{theorem}
\begin{proof}[Proof (reproduced for completeness)]
First, $\alpha_t\in[-\gamma,1+\gamma]$ for all $t$, by induction. If $\alpha_t<0$, then $\mathrm{err}_t=0$ and $\alpha_{t+1}=\alpha_t+\gamma\alpha\in(\alpha_t,\gamma\alpha)$. If $\alpha_t>1$, then $\mathrm{err}_t=1$ and $\alpha_{t+1}=\alpha_t-\gamma(1-\alpha)\in(1-\gamma(1-\alpha),\alpha_t)$. If $\alpha_t\in[0,1]$, then $\alpha_{t+1}\in[\alpha_t-\gamma(1-\alpha),\alpha_t+\gamma\alpha]$. In each case the interval lies in $[-\gamma,1+\gamma]$. Telescoping the update gives $\alpha_{T+1}-\alpha_1=\gamma\sum_{t\le T}(\alpha-\mathrm{err}_t)$, hence $|\frac1T\sum\mathrm{err}_t-\alpha|=|\alpha_{T+1}-\alpha_1|/(\gamma T)\le\max\{1+\gamma-\alpha_1,\alpha_1+\gamma\}/(\gamma T)$.
\end{proof}
\begin{theorem}[Kelly shrinkage under estimation risk; classical, included for completeness]\label{thm:kelly}
Let an asset have excess drift $\mu$ and volatility $\sigma$, with continuous rebalancing at a constant fraction $f$, so that the long-run excess log-growth rate is $g(f)=f\mu-\tfrac12f^2\sigma^2$ \citep{merton1969}. Let $\hat\mu$ be an estimate with $\E\hat\mu=\mu$ and $\Var\hat\mu=s^2$, independent of future returns, and use $f=c\,\hat\mu/\sigma^2$. Then
\[
  \E\,g=\frac{c\mu^2-\tfrac12c^2(\mu^2+s^2)}{\sigma^2},\qquad
  c^\star=\frac{\mu^2}{\mu^2+s^2}=\frac{\mathrm{SNR}}{1+\mathrm{SNR}},\quad\mathrm{SNR}\coloneqq\frac{\mu^2}{s^2},
\]
with optimal growth $\mu^4/\big(2\sigma^2(\mu^2+s^2)\big)$. Full Kelly ($c=1$) yields $(\mu^2-s^2)/(2\sigma^2)$, which is negative whenever $s>|\mu|$.
\end{theorem}
\begin{proof}
$\E[f\mu]=c\mu^2/\sigma^2$ and $\E[f^2]\sigma^2/2=c^2(\mu^2+s^2)/(2\sigma^2)$. The resulting concave quadratic in $c$ is maximised at $c^\star$, and substitution gives the remaining expressions.
\end{proof}

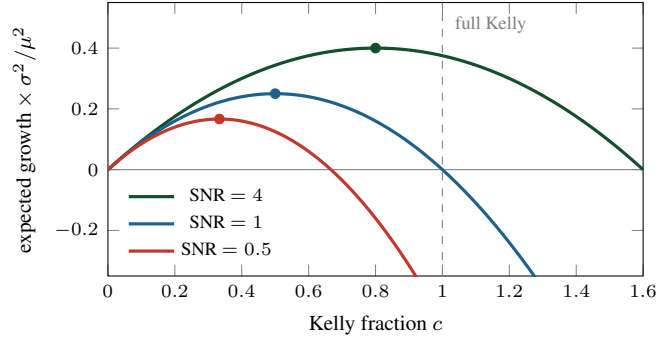
\begin{figure}[ht]
\centering
\begin{tikzpicture}
\begin{axis}[width=0.62\textwidth, height=5.2cm, xmin=0, xmax=1.6, ymin=-0.35, ymax=0.55,
  xlabel={\footnotesize Kelly fraction $c$},
  ylabel={\footnotesize expected growth $\times\,\sigma^2/\mu^2$},
  legend style={font=\scriptsize, draw=none, fill=none, at={(0.02,0.02)}, anchor=south west},
  tick label style={font=\scriptsize}, samples=100, domain=0:1.6]
\addplot[gray, thin, forget plot] coordinates {(0,0) (1.6,0)};
\addplot[good!60!black, very thick] {x - 0.5*x^2*(1+1/4)}; \addlegendentry{SNR $=4$}
\addplot[do, very thick] {x - 0.5*x^2*(1+1/1)}; \addlegendentry{SNR $=1$}
\addplot[echo, very thick] {x - 0.5*x^2*(1+1/0.5)}; \addlegendentry{SNR $=0.5$}
\addplot[only marks, mark=*, mark size=1.8pt, good!60!black, forget plot] coordinates {(0.8,0.4)};
\addplot[only marks, mark=*, mark size=1.8pt, do, forget plot] coordinates {(0.5,0.25)};
\addplot[only marks, mark=*, mark size=1.8pt, echo, forget plot] coordinates {(0.3333,0.16667)};
\addplot[dashed, gray, forget plot] coordinates {(1,-0.35) (1,0.55)};
\node[font=\scriptsize, text=gray, anchor=south west] at (axis cs:1.01,0.42) {full Kelly};
\end{axis}
\end{tikzpicture}
\caption{\Cref{thm:kelly}. Expected growth as a function of the Kelly fraction for three signal-to-noise ratios; dots mark $c^\star=\mathrm{SNR}/(1+\mathrm{SNR})$. With noisy edge estimates, full Kelly destroys growth. The conformal interval width of \cref{thm:aci} supplies $s$.}
\label{fig:kelly}
\end{figure}

\begin{proposition}[Well-posed portfolio construction]\label{prop:convex}
For $\hat\Sigma\succ0$, risk aversion $\gamma_r>0$, cost coefficients $c_i,\eta_i\ge0$, factor exposures $B$ and bounds $\bar w,G>0$, the programme
\begin{align*}
  \max_{\w}\quad & \bm\alpha^\top\w-\tfrac{\gamma_r}{2}\w^\top\hat\Sigma\w-\sum_ic_i|w_i-w_i^0|-\sum_i\eta_i\sigma_i\frac{|w_i-w_i^0|^{3/2}}{\sqrt{V_i}}\\
  \text{s.t.}\quad & B^\top\w=\bm 0,\quad \bm 1^\top\w=0,\quad |w_i|\le\bar w,\quad \|\w\|_1\le G
\end{align*}
has a unique solution. The $3/2$-power term integrates square-root impact \citep{bouchaud2018}.
\end{proposition}
\begin{proof}
The feasible set is a nonempty ($\w=\bm 0$ is feasible), compact and convex polytope. The objective is continuous and strictly concave: it is a linear term, minus a strictly convex quadratic, minus convex functions $|\cdot|$ and $|\cdot|^{3/2}$. A continuous function attains its maximum on a compact set, and strict concavity makes the maximiser unique.
\end{proof}

\section{Numerical verification}\label{app:verify}\label{sec:verify}
Every closed form in this paper was checked against an independent computation (\cref{tab:verify}). Theory values are exact functions of the model parameters. The Gaussian results of \cref{sec:model} were checked by regression on $2\times10^6$ simulated draws. \Cref{thm:identity} was checked with fat-tailed values ($t_3$) and Laplace message noise, which tests its distribution-free claim. Optimal speech was compared with numerical optimisation, \cref{thm:semantic} with brute-force enumeration of every parent labelling, and \cref{thm:levy} with a smooth Gaussian process with a known lag. Both forms of \cref{thm:robustnn} were checked on $3{,}000$ random pairs of target and model distributions: no certified pair was ever misordered, and the exact form certified $3.6$ times as many pairs as the Pinsker form (the ratio depends on how instances are drawn; see \texttt{scripts/check\_certificate.py}). The $t_3$ exaggeration case converges most slowly because the fourth moment of a $t_3$ variable is infinite; the pathwise identity $r=\lambda x-\varphi\varepsilon$ itself holds to machine precision.

\begin{table}[ht]
\centering
\caption{Theory versus independent checks ($2\times10^6$ draws, seed 0).}
\label{tab:verify}
\small
\begin{tabularx}{0.94\textwidth}{@{}l>{\raggedright\arraybackslash}Xrr@{}}
\toprule
Result & Quantity & Theory & Check \\
\midrule
\multicolumn{4}{@{}l}{\footnotesize\color{gray} Three-voice economy (\cref{sec:model})}\\
\cref{cor:echo} & Bayes weight on the echo channel & $0$ & $-0.0006$ \\
\cref{thm:followfade} & return loading on the echo channel & $-0.1840$ & $-0.1843$ \\
\cref{thm:saydo} & loading on positioning, $c_x$ & $0.4012$ & $0.4004$ \\
\cref{thm:absorption} & absorption coefficient, $c_A$ & $-0.1997$ & $-0.2004$ \\
\multicolumn{4}{@{}l}{\footnotesize\color{gray} Strategic speech (\cref{sec:strategic})}\\
\cref{thm:speech} & speech slope $\psi$ ($\varphi=0.6$, $a=0.5$) & $-0.7143$ & $-0.7143$ \\
\cref{thm:speech} & trade slope $\chi$ ($\varphi=0.6$, $a=0.5$) & $1.4286$ & $1.4286$ \\
\cref{thm:identity} & $\Cov(r,\text{Say})$, $t_3$ values, false alarm & $-2.294$ & $-2.293$ \\
\cref{thm:identity} & $\Cov(r,\text{Say})$, $t_3$ values, exaggeration & $-7.608$ & $-7.676$ \\
\cref{thm:cycles} & 2-cycle points at $a=1.95$ & $\{0.8189,\,1.1311\}$ & $\{0.8189,\,1.1311\}$ \\
\cref{thm:cycles} & 2-cycle multiplier at $a=1.8$ & $0$ & $9\times10^{-16}$ \\
\multicolumn{4}{@{}l}{\footnotesize\color{gray} Measurement (\cref{sec:echo,sec:embeddings,sec:leadlag})}\\
\cref{thm:semantic} & largest error of the factorised posterior & $0$ & $1.1\times10^{-16}$ \\
\cref{thm:levy} & L\'evy-area rate, lag $0.5$ & $0.2356$ & $0.2360$ \\
\cref{thm:robustnn} & wrongly ordered certified pairs, $3{,}000$ random cases & $0$ & $0$ \\
\multicolumn{4}{@{}l}{\footnotesize\color{gray} Sizing (\cref{sec:sizing})}\\
\cref{thm:kelly} & growth at the shrunk Kelly fraction $c^\star$ & $0.01056$ & $0.01055$ \\
\cref{thm:kelly} & growth at full Kelly & $-0.0300$ & $-0.0300$ \\
\bottomrule
\end{tabularx}
\end{table}

Exactness does not by itself show that meaning \emph{helps}. We therefore simulated cascades in which each echo is a von Mises--Fisher copy of its parent with concentration $\varkappa$ (\cref{ass:marked}; Hawkes $\mu=0.6$, $\alpha=1$, $\beta=1.25$; $16$-dimensional embeddings; about $2{,}000$--$2{,}800$ articles per run) and asked each method to recover the true parent of every article. Timing alone recovers about one parent in four at every concentration. Adding meaning raises accuracy to $95\%$ when echoes are close copies, and it nearly eliminates the error in each article's probability of being an echo (\cref{fig:decluster}).

\begin{figure}[ht]
\centering
\begin{minipage}{0.48\textwidth}
\centering
\begin{tikzpicture}
\begin{semilogxaxis}[width=6.6cm, height=5.0cm, xmin=1.6, xmax=100, ymin=0, ymax=1,
  xtick={2,5,10,20,40,80}, xticklabels={2,5,10,20,40,80}, log ticks with fixed point,
  xlabel={\footnotesize echo concentration $\varkappa$}, ylabel={\footnotesize parent accuracy},
  tick label style={font=\footnotesize}, axis lines=left,
  legend style={font=\scriptsize, draw=none, fill=none, at={(0.99,0.31)}, anchor=south east}]
\addplot[gray, thick, mark=*, mark size=1.6pt] coordinates {(2,0.270) (5,0.268) (10,0.249) (20,0.211) (40,0.257) (80,0.212)};
\addlegendentry{timing only}
\addplot[echo, thick, mark=*, mark size=1.6pt] coordinates {(2,0.312) (5,0.452) (10,0.631) (20,0.844) (40,0.941) (80,0.953)};
\addlegendentry{timing + meaning}
\end{semilogxaxis}
\node[font=\bfseries] at (-0.9,4.3) {a};
\end{tikzpicture}
\end{minipage}\hfill
\begin{minipage}{0.48\textwidth}
\centering
\begin{tikzpicture}
\begin{semilogxaxis}[width=6.6cm, height=5.0cm, xmin=1.6, xmax=100, ymin=0, ymax=0.32,
  xtick={2,5,10,20,40,80}, xticklabels={2,5,10,20,40,80},
  ytick={0,0.1,0.2,0.3}, yticklabels={0,0.1,0.2,0.3},
  xlabel={\footnotesize echo concentration $\varkappa$}, ylabel={\footnotesize mean $|\,\Prob(\text{echo})-\ind\{\text{echo}\}\,|$},
  tick label style={font=\footnotesize}, axis lines=left]
\addplot[gray, thick, mark=*, mark size=1.6pt] coordinates {(2,0.272) (5,0.281) (10,0.252) (20,0.237) (40,0.289) (80,0.239)};
\addplot[echo, thick, mark=*, mark size=1.6pt] coordinates {(2,0.268) (5,0.258) (10,0.177) (20,0.078) (40,0.014) (80,0.001)};
\end{semilogxaxis}
\node[font=\bfseries] at (-0.9,4.3) {b};
\end{tikzpicture}
\end{minipage}
\caption{Meaning makes echo detection work (\cref{thm:semantic}). \textbf{a}: accuracy of the most likely parent. \textbf{b}: mean absolute error of each article's echo probability. Grey: timing only ($\varkappa=0$ in the posterior). Red: timing and meaning.}
\label{fig:decluster}
\end{figure}
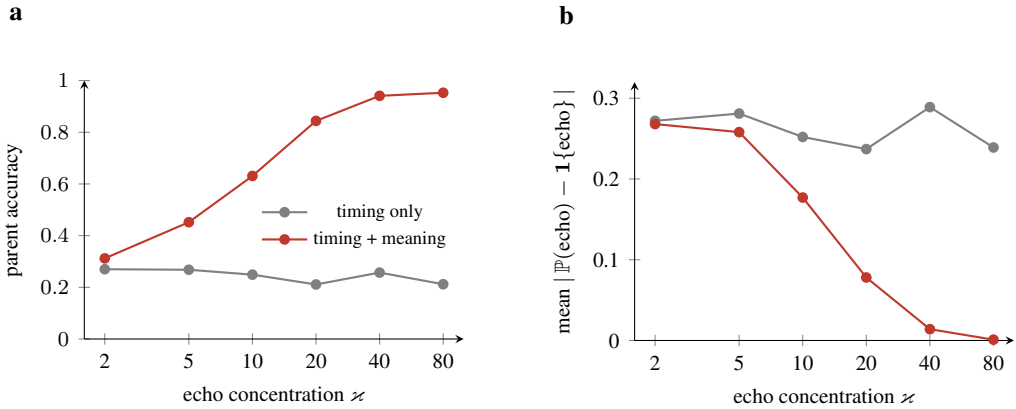

\section{Experimental details}\label{app:experiments}
\paragraph{The market.} There are $300$ firms and $60$ events per firm. Each firm hosts one institution whose parameters are fixed across events. A quarter of the institutions are non-strategic: they report value and do not trade. The rest are strategic, with credulity $\varphi$ and deterrence $a=2\lambda k$ drawn inside the shading, false-alarm or exaggeration region of \cref{cor:taxonomy}. Each firm's regime is drawn with probability one quarter for each of the four. At each event:
\begin{enumerate}[label=(\arabic*)]
  \item value is drawn, $v\sim\N(0,1)$;
  \item the institution trades $x=\chi v$, spread evenly over the ten days before it speaks, and then publishes $\text{Say}=\psi v+\N(0,0.3^2)$, with $(\psi,\chi)$ from \cref{thm:speech};
  \item the media publish a Hawkes cascade over the next ten days. The statement is the first article. Independent news items arrive at rate $0.4$ per day, each carrying the sentiment $v+\N(0,1.5^2)$. Every article can be echoed, with branching ratio $n=1-\varphi_0/\varphi$ so that the crowd's effective credulity is $\varphi$ (\cref{prop:virality}). Echoes arrive after exponential delays with rate $1.5$, copy their parent's sentiment up to noise $\N(0,0.15^2)$, and copy its $16$-dimensional embedding up to a von Mises--Fisher perturbation with $\varkappa=30$ (\cref{ass:marked});
  \item the crowd prices every article, $p=\sum_j w_j s_j+\lambda(x+u)$, where $w_j=\varphi_0$ for the statement and its echoes, $w_j=0.04$ for news and its echoes, and $u\sim\N(0,0.4^2)$ is noise trading;
  \item the forward return is $r=v-p+\eta$, where $\eta\sim\N(0,10^2)$ is information that arrives later, and positioning is revealed with error, $\text{Do}=x+\N(0,0.5^2)$.
\end{enumerate}

\paragraph{What the methods see.} Every method sees the article times, embeddings and sentiment scores, the statement, the noisy positioning, the price reaction, and daily paths of positioning and cumulative tone from ten days before the statement to ten days after it. No forecasting method sees value, regimes, parentage or any model parameter. The one exception is the false-alarm detector of Q2, which is trained on the true regime labels of the training half, as a supervised detector would be on labelled historical episodes.

\paragraph{Protocol.} The first $30$ events of each firm form the training set and the last $30$ the test set. Every estimated quantity uses the training half only. This includes the Hawkes parameters (pooled maximum likelihood), the echo concentration $\varkappa$ (marked likelihood), the narrative-implied price move, all regression and boosting weights, and all standardisations. The features of an event use only that event's articles and price reaction and the firm's earlier events. \Cref{tab:features} lists the features.

\begin{table}[ht]
\centering
\caption{Features built from the theory. All are computed causally.}
\label{tab:features}
\small
\begin{tabularx}{\textwidth}{@{}l>{\raggedright\arraybackslash}Xl@{}}
\toprule
Feature & Construction & Source \\
\midrule
News, echo & $s^{\text{new}}=\sum_j\Prob(\text{news}_j)\,s_j$ over articles other than the statement; $s^{\text{echo}}=\sum_j\Prob(\text{echo}_j)\,s_j$ & \cref{prop:mmse,thm:semantic} \\
Say--Do switch & correlation of Say and Do over the firm's previous $20$ events; words are faded when it is below $-0.2$ & \cref{thm:identity} \\
Absorption & price reaction minus the narrative-implied move, fitted linearly on Say, $s^{\text{new}}$ and $s^{\text{echo}}$ (pooled, or per firm; tone is their sum) & \cref{thm:absorption} \\
Lead--lag & L\'evy area of (positioning, $-$tone) over the event window, signed by the direction of the positioning change & \cref{thm:levy} \\
Composite & unit weights on the standardised terms above, with the switch applied to Say & \cref{sec:synthesis} \\
\bottomrule
\end{tabularx}
\end{table}

\paragraph{Calibration, and what we changed.} The features, tests and strategies were fixed before we looked at any of their results. The simulator's scale was then calibrated as follows. A pilot run exposed two degenerate properties of the raw simulator: prices moved about $3.3$ times as much as fundamentals, and fading the tone alone reached an information coefficient of $0.89$. We lowered the per-article impact of news from $0.12$ to $0.04$ and restricted the credulity of non-strategic institutions to $(0.4,0.95)$. We then chose the scale of the late information $\eta$ so that a linear regression on the raw voices, which uses no theory at all, reaches an out-of-sample information coefficient of about $0.17$. That is far more predictable than any real market. We chose it so that differences between methods are visible with $300$ firms, and only comparisons between methods should be read from the results. The calibration also caps what any method can achieve. Forecasting with the true mispricing $v-p$ gives an IC of $0.181$ (between $0.165$ and $0.206$ across seeds), only about $0.01$ above the linear baseline.

\paragraph{Metrics.} For forecasts we report the cross-sectional information coefficient (IC) in each of the $30$ test periods, its average and its $t$-statistic across periods, and the per-period Sharpe ratio of a dollar-neutral portfolio with weights proportional to the standardised score. The theory tests are pooled regressions of test-period returns on standardised features, with standard errors clustered by firm. The main market is simulated with five seeds, every robustness variant with three, and every sample size with two. Together these give $29$ distinct simulated markets. The ``timing only'' variant re-analyses three of them, so there are $32$ analyses in all.

\paragraph{Compute and software.} All experiments ran on a cloud container with $2$ CPU cores and $7$\,GB of memory; no GPU was used. One main-market seed (simulation, Hawkes fit, declustering of about $25{,}000$ articles, all strategies and tests) takes about $40$ seconds, and one text-experiment run takes about $4.5$ minutes on one core. The full suite of \cref{sec:experiments}, including robustness, sample-size, text and diagnostic runs, takes roughly one hour. We used NumPy and SciPy (BSD-3-Clause), scikit-learn \citep{pedregosa2011} (BSD-3-Clause) and autograd (MIT). Versions: NumPy 2.4.4 \citep{harris2020}, SciPy 1.17.1 \citep{virtanen2020}, scikit-learn 1.8.0. Pilot runs used for calibration took about as much compute as one main-market seed each. The code, released under the MIT licence with tests for every theorem, regenerates every table and figure. Comments in the code refer to the theorem numbering of an extended version of this paper.

\paragraph{Hyperparameters.} \emph{Text benchmark:} the projection $W\in\R^{V\times16}$ is initialised with independent $\N(0,1/V)$ entries and trained with Adam (learning rate $0.05$, $\beta_1=0.9$, $\beta_2=0.999$) for $600$ steps on random batches of $256$ headlines; $\tau=0.1$ was fixed in advance and not tuned; the bandwidth $b$ is the median pairwise outcome distance over $400$ training headlines, and $0.3b$ in the sharp-target run; analog forecasts use a Gaussian kernel with bandwidth $0.05$ on embedding distance; the certificate is evaluated on $150$ random training anchors. The ridge baseline uses $\alpha=1$, and gradient boosting uses scikit-learn's histogram-based regressor with $300$ iterations and learning rate $0.05$. \emph{Simulated market:} both boosting models use the same settings; the false-alarm detector is an $\ell_2$-regularised logistic regression ($C=1$) on standardised features; the Say--Do window is $W=20$ past events (at least $5$) and the switch threshold is $-0.2$. None of these values was tuned on test data.

\section{Additional results}\label{app:results}
\paragraph{Q1: repetition can be separated from news, and fading it works.} The semantic declustering step recovers the echo concentration almost exactly ($\hat\varkappa=30.27\pm0.07$ against a true $30$) and the share of articles that are echoes ($0.422$ against $0.427$). On the same simulated markets, declustering by timing alone estimates an echo share of $0.368$ against a true $0.430$, because it mistakes many echoes for news. The prediction of \cref{cor:echo} holds without exception. Echo sentiment predicts returns with a significantly negative sign in every one of the $29$ simulated markets, with a mean $t$-statistic of $-8.0$ in the main market (\cref{tab:tests}). Meaning sharpens the test: on identical markets, the $t$-statistic weakens from $-7.5$ to $-6.0$ when timing alone is used. The news prediction of \cref{cor:news} also holds, but more weakly. The sign is positive in every market, the mean $t$-statistic is $+2.1$, and the effect is significant in $21$ of the $29$ markets. News is harder to detect because each news article is a noisy signal of value and the crowd underreacts to it only modestly.

\begin{table}[ht]
\centering
\caption{Theory tests on the test half of the main market (five seeds). Coefficients are in units of fundamental value per training standard deviation of the feature. $t$: firm-clustered, averaged over seeds. The last column counts seeds with $|t|>2$ and the predicted sign.}
\label{tab:tests}
\small
\begin{tabularx}{\textwidth}{@{}>{\raggedright\arraybackslash}Xlrrc@{}}
\toprule
Regressor (controls) & Prediction & Coef. & $t$ & As predicted \\
\midrule
News sentiment $s^{\text{new}}$ (Say, Do) & $+$ \ (\cref{cor:news}) & $+0.35$ & $+2.1$ & 3/5 \\
Echo sentiment $s^{\text{echo}}$ (Say, Do) & $-$ \ (\cref{cor:echo}) & $-1.52$ & $-8.0$ & 5/5 \\
\addlinespace
Say, fade group (Do) & $-$ \ (\cref{thm:identity}) & & $-3.3$ & 4/5 \\
Say, fade group (Do, news, echo) & $-$ & $-0.06$ & $-0.3$ & 1/5 \\
Say, follow group (Do, news, echo) & either & $+0.10$ & $+0.4$ & -- \\
Do (Say by group, news, echo) & $+$ & $+0.14$ & $+0.9$ & 1/5 \\
\addlinespace
Absorption, all events & either \ (\cref{thm:absorption}) & $-1.27$ & $-11.6$ & -- \\
\quad institution does not trade & $-$ & $-1.25$ & $-3.6$ & 4/5 \\
\quad institution trades & either & $-1.27$ & $-11.1$ & -- \\
Absorption, firm-level multiplier & either & $-0.72$ & $-9.0$ & -- \\
\quad institution does not trade & $-$ & $-1.03$ & $-4.0$ & 5/5 \\
\addlinespace
Signed L\'evy area & $+$ \ (\cref{thm:levy}) & $+0.23$ & $+1.0$ & 1/5 \\
\bottomrule
\end{tabularx}
\end{table}

\paragraph{Q2: the Say--Do covariance finds the institutions whose words should be faded.} On its own, the rolling correlation between Say and Do separates the events of false-alarm institutions from all others with an AUC of $0.895$ (\cref{fig:experiments}a). A logistic classifier that adds four other observable features, trained on the true labels of the training half, reaches $0.934\pm0.010$ on later events of the same firms. Real data would supply such labels only through episodes identified after the fact. The single-feature AUCs are reported as $\max(\text{AUC},1-\text{AUC})$, so the values near $0.5$ for absorption and echo share mean ``uninformative''. Used as the switch of \cref{sec:synthesis}, the same correlation labels correctly whether a firm's words should be faded (false-alarm and exaggerating institutions, the two regimes in which $\Cov(\text{Say},\text{Do})<0$) $85\%$ of the time from five past events and $96\%$ of the time from forty (\cref{fig:experiments}b). In the group it flags, statements predict returns with the sign reversed: the average $t$-statistic is between $-3.0$ and $-3.3$ at every window length, and it is significant in four of five seeds. In the other group the average $t$-statistic is between $-0.6$ and $-1.0$. We report that sign as an empirical finding. The identity does not cover non-strategic institutions, whose silence in the market is not optimal given the crowd's credulity.

Two caveats matter for real data. First, once news and echo sentiment are added as controls, the reversed coefficient on Say in the flagged group is no longer significant ($t=-0.3$). Echo sentiment alone is enough to remove it ($t=+0.2$ on average), while news alone only weakens it ($t=-2.0$). In this market most of a statement's price impact travels through its echoes, so the echo control absorbs it. Second, Do's own coefficient is imprecise, for two reasons. With Say, news and echo all controlled its $t$-statistic is $+0.9$, partly because Say and Do are nearly collinear within the strategic regimes (their correlation is $-0.87$ for false alarms, $-0.78$ for exaggeration and $+0.54$ for shading), so a regression cannot divide credit between them. Even with Say removed, the $t$-statistic averages only $+1.6$, because the institution's price impact is small next to the information that arrives later. The collinearity is also why \cref{thm:saydo} works with the gap between the two voices rather than with either voice alone.

\paragraph{Q3: absorption measures the narrative model as much as the market; lead--lag is weak.} Absorption predicts strong \emph{reversal} ($t=-11.6$), and it does so in all $29$ markets. It does so both where institutions trade, where \cref{thm:absorption} allows either sign, and where they do not ($t=-3.6$), where it predicts reversal. The second result looks like a confirmation, but it is not, and \texttt{experiments/diagnostics.py} shows why.

Absorption is the residual of the price reaction after the narrative-implied move has been removed. Here only $12$--$19\%$ of the residual's variance is the institution's trade and $2$--$4\%$ is noise trading. Almost all of the rest is \emph{narrative misfit}, the part of the crowd's reaction that a linear narrative model fails to capture. The misfit has a simple source. The statement's price impact is multiplied by the size of its echo cascade, and $s^{\text{echo}}$ pools the statement's echoes, which the crowd weights heavily, with the more numerous echoes of news, which it weights lightly. Adding Say times cascade size to the narrative model weakens the reversal from $t\approx-12$ to $t\approx-3.5$ on average. If we replace the estimate with the true absorption $\lambda(x+u)$, which a perfect narrative model would recover, the coefficient is insignificant in both groups: $t$ ranges from $-1.2$ to $+0.5$ where institutions trade and from $-1.4$ to $+0.8$ where they do not. \Cref{lem:robust} protects absorption against a wrong multiplier. It does not protect it against a price that is not linear in the channels the model uses.

Our experiments therefore neither confirm nor refute \cref{thm:absorption}: the simulated trades are too small relative to later information to reveal either sign. What they do show is practical. The sign of estimated absorption depends on the narrative model and must be estimated, never assumed. In the unit-weight composite, the absorption term enters with the wrong sign for this market, and removing it doubles the composite's IC from $0.049$ to $0.096$.

The signed L\'evy area has the predicted positive sign in $27$ of the $29$ markets, but it is significant in only $9$ ($t=+1.0$ in the main market). It is more useful for detecting false alarms (AUC $0.63$) than for predicting returns.

\paragraph{Q4: the theory adds little forecasting power, but there was little to add.} \Cref{fig:ic} ranks twelve forecasts. Following the tone loses money. Following institutional statements loses money too, because half of the institutions are in regimes whose words should be faded, and even the others' statements are already priced through their echoes. Fading the price reaction is a strong baseline (IC $0.162$). A linear regression on the raw voices does slightly better ($0.169$). Adding the theory's features to that regression changes almost nothing ($0.171$, at most about $0.005$ better in any seed). For gradient boosting the features help a little: $0.109$ becomes $0.118$ in the main market, and the features improve boosting by more than $0.001$ in $18$ of the $29$ markets. The unit-weight composite of \cref{sec:synthesis} is weak ($0.049$), mostly because of the absorption term discussed above.

This comparison has little power, and we do not read much into it. Forecasting with the true mispricing $v-p$, which no method can observe, reaches only $0.181$, so there is about $0.01$ of IC between the linear baseline and perfection. With so little room, the result is consistent with the theory adding nothing and also with it adding a great deal of what is left. Two points are nonetheless clear. First, \cref{thm:followfade} says the optimal forecast is linear in the voices, and a regression free to learn the weights comes close to the ceiling, as that theorem suggests. Second, the unit-weight composite is not a good trading rule; the composite of \cref{sec:synthesis} should be used with fitted weights. The theory's clearest contributions in these experiments are the signs it predicts and the diagnostics of Q1 and Q2, not extra forecasting power.

\begin{figure}[ht]
\centering
\begin{tikzpicture}
\begin{axis}[xbar, width=0.62\textwidth, height=8.2cm, bar width=7pt,
  xmin=-0.15, xmax=0.21, ymin=0.4, ymax=12.6, y dir=reverse,
  ytick={1,...,12},
  yticklabels={Linear: raw + theory, Linear: raw voices, Fade the price reaction, Fade the tone, Boosting: raw + theory, Boosting: raw voices, News minus echo, Theory composite (unit weights), Follow Do, Composite with firm-level absorption, Follow Say, Follow the tone},
  xtick={-0.1,-0.05,0,0.05,0.1,0.15,0.2},
  xticklabel style={/pgf/number format/fixed, /pgf/number format/precision=2},
  xlabel={\footnotesize out-of-sample information coefficient (mean $\pm$ s.d., five seeds)},
  tick label style={font=\footnotesize}, axis lines=left,
  legend style={font=\scriptsize, draw=none, fill=none, at={(0.99,0.02)}, anchor=south east}, legend cell align=left]
\addplot[fill=do!70, draw=none, bar shift=0pt, error bars/.cd, x dir=both, x explicit]
  coordinates {(0.171,1) +- (0.018,0) (0.169,2) +- (0.017,0) (0.118,5) +- (0.019,0) (0.109,6) +- (0.015,0)};
\addlegendentry{fitted models}
\addplot[fill=say!65, draw=none, bar shift=0pt, error bars/.cd, x dir=both, x explicit]
  coordinates {(0.091,7) +- (0.018,0) (0.049,8) +- (0.010,0) (0.027,10) +- (0.012,0)};
\addlegendentry{theory, unit weights}
\addplot[fill=gray!55, draw=none, bar shift=0pt, error bars/.cd, x dir=both, x explicit]
  coordinates {(0.162,3) +- (0.018,0) (0.119,4) +- (0.013,0) (0.035,9) +- (0.011,0) (-0.047,11) +- (0.006,0) (-0.119,12) +- (0.013,0)};
\addlegendentry{single-voice rules}
\draw[black] (axis cs:0,0.4) -- (axis cs:0,12.6);
\end{axis}
\end{tikzpicture}
\caption{Out-of-sample forecasts in the main simulated market. The fitted models use the training half. ``Raw voices'' are Say, Do, tone, article count and the price reaction; ``theory'' adds the features of \cref{tab:features}.}
\label{fig:ic}
\end{figure}
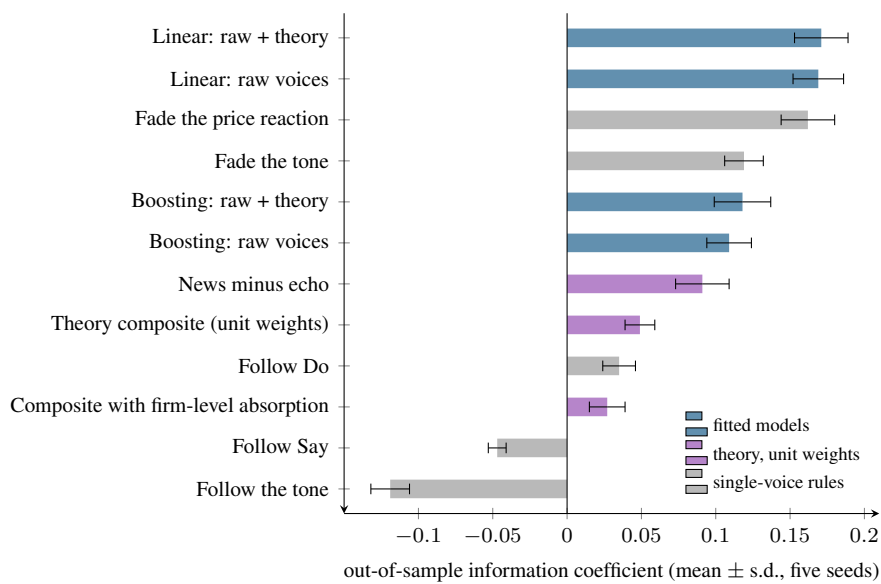

\paragraph{Robustness.} \Cref{tab:robust} changes one assumption at a time on three new seeds. Fat-tailed values and power-law echo delays (a misspecified kernel) leave the echo, news and detection results intact. Partial disclosure, in which only $30\%$ of the trade is disclosed and positioning is measured with twice the noise, leaves the echo and news results intact but lowers the detection AUC to $0.857$. A saturating crowd, which prices $1.5\tanh(s/1.5)$ instead of $s$, halves the IC of every forecast, but detection is unaffected and the echo and news signs survive. The controlled test of reversed words (\cref{tab:tests}, fourth row) does not: under the saturating crowd its $t$-statistic has the wrong sign in all three seeds. When meaning is weak ($\varkappa=5$) or ignored, the echo share is underestimated by about six percentage points and the echo test weakens, as \cref{thm:semantic} leads one to expect. More events per firm strengthen the echo test, somewhat more slowly than the square root of the sample: over two seeds the $t$-statistic is $-5.4$, $-6.9$ and $-9.3$ with $20$, $40$ and $80$ events, while detection stays at an AUC of $0.93$--$0.94$.

\begin{table}[ht]
\centering
\caption{Robustness: one change at a time, three seeds each. ``Reference'' is the main market on the same seeds. IC: out-of-sample information coefficient. Echo share: estimated (true).}
\label{tab:robust}
\small
\setlength{\tabcolsep}{3.6pt}
\begin{tabular}{@{}lrrrrrrrc@{}}
\toprule
& \multicolumn{3}{c}{IC} & \multicolumn{3}{c}{$t$-statistic} & & \\
\cmidrule(lr){2-4}\cmidrule(lr){5-7}
Variant & fade price & linear & + theory & echo & news & L\'evy & AUC & echo share \\
\midrule
Reference & $0.184$ & $0.187$ & $0.190$ & $-7.5$ & $+3.6$ & $+3.6$ & $0.941$ & $0.425$ ($0.430$) \\
Fat tails ($t_3$ values) & $0.134$ & $0.137$ & $0.141$ & $-8.1$ & $+2.4$ & $+0.5$ & $0.924$ & $0.423$ ($0.428$) \\
Saturating crowd & $0.090$ & $0.099$ & $0.101$ & $-5.9$ & $+3.2$ & $-0.0$ & $0.942$ & $0.425$ ($0.430$) \\
Power-law echo delays & $0.171$ & $0.176$ & $0.177$ & $-7.9$ & $+2.1$ & $+1.9$ & $0.939$ & $0.417$ ($0.429$) \\
30\% of trade disclosed & $0.184$ & $0.183$ & $0.191$ & $-7.5$ & $+4.1$ & $+2.9$ & $0.857$ & $0.425$ ($0.430$) \\
Weak meaning ($\varkappa=5$) & $0.173$ & $0.187$ & $0.188$ & $-7.2$ & $+1.9$ & $+2.2$ & $0.941$ & $0.373$ ($0.429$) \\
Timing only (reference markets) & $0.184$ & $0.187$ & $0.189$ & $-6.0$ & $+2.6$ & $+3.7$ & $0.940$ & $0.368$ ($0.430$) \\
\bottomrule
\end{tabular}
\end{table}

\paragraph{Q5: return-aligned embeddings organise text by consequence.}\label{app:text} For the embedding results we generated a corpus of short synthetic headlines. Each headline mixes four topic words, two tone words, one cue that decides whether the move is \emph{overdone} (it reverses within a week) or \emph{fundamental} (it continues), three filler words and a company code. Topic words are numerous and have nothing to do with outcomes, so a generic bag-of-words geometry groups headlines by topic. Each headline has a three-part outcome: the next-day move, the week-ahead move and the change in implied volatility. We trained a $16$-dimensional linear projection of TF-IDF vectors with the loss of \cref{thm:racl} ($\tau=0.1$, bandwidth by the median heuristic, $600$ Adam steps on batches of $256$) on $2{,}000$ headlines, and tested it on $1{,}000$ new ones.

Training moves the geometry from topic to consequence (\cref{tab:text,fig:text}a). In TF-IDF space the ten nearest neighbours of a test headline share its topic $85\%$ of the time and its consequence (tone and cue together) $51\%$ of the time. After training they share its consequence $100\%$ of the time and its topic $34\%$ of the time, close to the chance level of $25\%$. The analog forecast of \cref{prop:analog} predicts the week-ahead move with an IC of $0.892$, against $0.192$ for TF-IDF neighbours and a ceiling of $0.936$ set by the outcome noise. Linear regression on the same words fails completely ($-0.02$), because the outcome depends on the \emph{interaction} of tone and cue, which a linear model of word counts cannot express. Gradient boosting captures the interaction ($0.837$) but falls short of the embedding. That comparison favours the embedding, which is trained on all three outcomes: the next-day move reveals the tone and the volatility change reveals the cue, while boosting sees only the week-ahead target.

The certificate is the negative result (\cref{fig:text}b). After training, the mean excess loss per anchor is $0.021$--$0.026$ across seeds. That is small, but target probabilities over $149$ neighbours are of order $1/149\approx0.007$, and neither form of \cref{thm:robustnn} certifies a single comparison. Two things stand in the way. The median-heuristic bandwidth makes the targets almost flat, and the part of each outcome that the text cannot predict leaves an excess loss that no embedding can remove. When we remove both obstacles, using noise-free outcomes and a bandwidth of $0.3$ times the median, the exact form (ii) certifies $45\%$ of the ordered comparisons with no errors, while the Pinsker form (i) still certifies none. With the sharp bandwidth and realistic outcome noise, the excess loss stays at $0.14$ and nothing is certified, even though the analog IC rises to $0.931$, close to its ceiling. The certificate never made a wrong claim in any run. It becomes informative, however, only when outcomes are close to predictable from the text, which real returns are not.

\begin{table}[ht]
\centering
\caption{Text experiment: $1{,}000$ test headlines, mean $\pm$ s.d.\ over three seeds. Neighbour columns: share of the ten nearest training headlines with the same consequence (tone and cue) or the same topic; chance is $0.25$ for both. The ceiling for the IC is $0.936$.}
\label{tab:text}
\small
\begin{tabular}{@{}lrrr@{}}
\toprule
Method & Week-ahead IC & Same consequence & Same topic \\
\midrule
TF-IDF neighbours (generic) & $0.192\pm0.003$ & $0.506$ & $0.846$ \\
Random 16-d projection (generic) & $0.039\pm0.016$ & $0.327$ & $0.622$ \\
Ridge regression on words (supervised) & $-0.019\pm0.016$ & -- & -- \\
Gradient boosting on words (supervised) & $0.837\pm0.008$ & -- & -- \\
Return-aligned embedding, 16-d & $\mathbf{0.892\pm0.009}$ & $\mathbf{1.000}$ & $0.344$ \\
\bottomrule
\end{tabular}
\end{table}

\begin{figure}[ht]
\centering
\begin{minipage}[t]{0.48\textwidth}
\centering
\begin{tikzpicture}
\begin{semilogxaxis}[width=6.6cm, height=5.0cm, xmin=0.8, xmax=800, ymin=0, ymax=1.05,
  xtick={1,10,100,600}, xticklabels={1,10,100,600},
  xlabel={\footnotesize training step}, ylabel={\footnotesize share or IC},
  tick label style={font=\footnotesize}, axis lines=left,
  legend style={font=\scriptsize, draw=none, fill=none, at={(0.5,-0.30)}, anchor=north}, legend cell align=left]
\addplot[do, very thick, mark=*, mark size=1.5pt] coordinates {(1,0.388) (25,0.529) (50,0.550) (100,0.662) (200,0.992) (400,1.000) (600,1.000)};
\addlegendentry{neighbours: same consequence}
\addplot[gray, very thick, mark=square*, mark size=1.4pt] coordinates {(1,0.554) (25,0.501) (50,0.438) (100,0.382) (200,0.339) (400,0.339) (600,0.344)};
\addlegendentry{neighbours: same topic}
\addplot[echo, very thick, mark=triangle*, mark size=1.8pt] coordinates {(1,0.139) (25,0.201) (50,0.295) (100,0.551) (200,0.878) (400,0.900) (600,0.892)};
\addlegendentry{analog IC}
\addplot[echo, dashed, domain=0.8:800, samples=2] {0.192};
\node[font=\scriptsize, text=echo, anchor=south east] at (axis cs:700,0.192) {TF-IDF IC};
\end{semilogxaxis}
\node[font=\bfseries] at (-0.9,4.1) {a};
\end{tikzpicture}
\end{minipage}\hfill
\begin{minipage}[t]{0.48\textwidth}
\centering
\begin{tikzpicture}
\begin{semilogxaxis}[width=6.6cm, height=5.0cm, xmin=0.8, xmax=800, ymin=-0.02, ymax=0.6,
  xtick={1,10,100,600}, xticklabels={1,10,100,600},
  ytick={0,0.1,0.2,0.3,0.4,0.5},
  yticklabel style={/pgf/number format/fixed, /pgf/number format/precision=1},
  xlabel={\footnotesize training step}, ylabel={\footnotesize certified share of comparisons},
  tick label style={font=\footnotesize}, axis lines=left,
  legend style={font=\scriptsize, draw=none, fill=none, at={(0.5,-0.30)}, anchor=north}, legend cell align=left]
\addplot[do, very thick, mark=*, mark size=1.5pt] coordinates {(1,0) (25,0) (50,0.0549) (100,0.4602) (200,0.4915) (400,0.4915) (600,0.4534)};
\addlegendentry{exact (ii), noise-free, sharp}
\addplot[say, very thick, dashed, mark=square*, mark size=1.3pt] coordinates {(1,0) (25,0) (50,0) (100,0) (200,0) (400,0) (600,0)};
\addlegendentry{Pinsker (i), same run}
\addplot[gray, thick, mark=triangle*, mark size=1.6pt, mark options={solid}] coordinates {(1,0) (25,0) (50,0) (100,0) (200,0) (400,0) (600,0)};
\addlegendentry{exact (ii), realistic noise}
\end{semilogxaxis}
\node[font=\bfseries] at (-0.9,4.1) {b};
\end{tikzpicture}
\end{minipage}
\caption{Return-aligned training. \textbf{a}: the ten nearest neighbours of test headlines switch from sharing a topic to sharing a consequence, and the analog forecast improves with them (mean of three seeds). \textbf{b}: share of ordered outcome comparisons certified by \cref{thm:robustnn} on $150$ training anchors. With noise-free outcomes and a bandwidth of $0.3$ times the median, the exact form certifies about half of them. The Pinsker form certifies none, and with realistic outcome noise neither form certifies any. No certified comparison was ever wrong.}
\label{fig:text}
\end{figure}
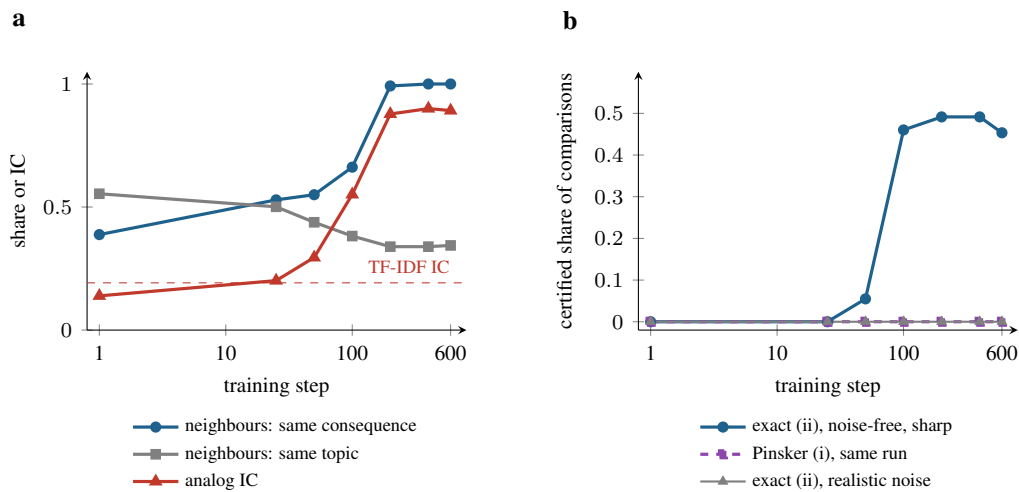

\section{Broader impacts and responsible use}\label{app:impact}

\paragraph{Intended use.} The framework is designed to help investors, researchers and supervisors read financial narratives critically, by asking whether words agree with revealed positions and whether a story is news or repetition. It uses public information and lawful disclosures only.

\paragraph{Interpretation.} A strategic reading means that narrative and capital moved in opposite directions in a particular order. It is not evidence of deception. Firms maintain information barriers between research and trading, and positioning can diverge from published views for legitimate reasons, such as hedging, client flow and differing horizons. Outputs should never be used to accuse individuals or firms without independent evidence.

\paragraph{Misuse.} Understanding how narratives move prices could in principle help someone design them. Creating or spreading narratives intended to move prices is market manipulation and illegal in most jurisdictions. We release no data, trained models or tools for generating text, and the simulator describes aggregate behaviour, not specific market participants.

\paragraph{Market effects.} If many participants faded the same echoes, the returns we describe would shrink, and crowding could create new risks. \Cref{thm:cycles} is itself a warning that reading rules and the texts they read co-evolve.

\end{document}